\documentclass{amsart}
\usepackage{latexsym}

\usepackage{amsmath,amsthm,amsfonts,amscd,eucal}
\usepackage{graphicx}
\usepackage{epsfig}
\usepackage{epstopdf} 

\numberwithin{equation}{section}

\def\cb{{\mathcal B}}

\def\cd{{\mathcal D}}

\def\cf{{\mathcal F}}

\def\ch{{\mathcal H}}

\def\car{{\mathcal R}}

\def\bc{{\mathbb C}}
\def\bg{{\mathbb G}}
\def\bh{{\mathbb H}}

\def\bn{{\mathbb N}}

\def\br{{\mathbb R}}

\def\bz{{\mathbb Z}}

\def\ga{{\mathfrak A}} 
\def\gb{{\mathfrak B}}

 \def\gph{{\mathfrak h}}

\def\gw{{\mathfrak W}}

\def\a{\alpha}
\def\b{\beta}
        \def\G{\Gamma}
\def\d{\delta}        \def\D{\Delta}
\def\eps{\varepsilon}

\def\th{\vartheta}

\def\l{\lambda}       \def\La{\Lambda}
\def\m{\mu}

\def\r{\rho}
\def\s{\sigma}        \def\S{\Sigma}
     
\def\t{\tau}

\def\f{\varphi} \def\F{\Phi}

\def\om{\omega}        \def\Om{\Omega}

\newtheorem{Thm}{Theorem}[section]

\newtheorem{Prop}[Thm]{Proposition}
\newtheorem{Lemma}[Thm]{Lemma}

\theoremstyle{definition}
\newtheorem{Dfn}[Thm]{Definition}

\newtheorem{Rem}[Thm]{Remark} 
\theoremstyle{remark}

\def\di{\mathop{\rm d}\!}

\def\im{\mathop{\rm Im}}
\def\Int{\mathop{\rm int}}

\def\itm#1{\item{$(#1)$}}

\newcommand{\ssv}{\text{span}\,}

\newcommand{\ccr}{\textrm{CCR}}
\newcommand{\tr}{\textrm{Tr}}

\def\idd{{\bf 1}\!\!{\rm I}}

\def\supp{\mathop{\rm supp}}

\newcommand{\nn}{\nonumber}

\begin{document}

\title[harmonic analysis and Bose Einstein condensation]
{harmonic analysis on Cayley Trees II: the Bose Einstein condensation}
\author{Francesco Fidaleo}
\address{Dipartimento di Matematica,
Universit\`{a} di Roma Tor Vergata, 
Via della Ricerca Scientifica 1, Roma 00133, Italy} 

\email{fidaleo@mat.uniroma2.it}


\date{\today}

\begin{abstract}

We investigate the Bose--Einstein Condensation on 
non homogeneous non amenable networks for the model describing arrays of 
Josephson junctions. The graphs under investigation are obtained by
adding density zero perturbations to the homogeneous Cayley Trees.
The resulting topological model, whose 
Hamiltonian is the pure hopping one given by the opposite of the 
adjacency operator, has also a mathematical interest in itself. The present paper is then the application to the Bose--Einstein Condensation phenomena, of the harmonic analysis aspects,  
previously investigated in a separate work, for such non amenable graphs.
Concerning the appearance of the Bose--Einstein Condensation, the results are surprisingly in accordance with the previous ones, despite the lack of amenability. The appearance of the hidden spectrum for low energies always implies that the critical density is finite for all the models under consideration. First we show that, even when the critical density is finite, if the adjacency operator of the graph is recurrent, it is impossible to exhibit temperature states which are locally normal (i.e. states for which the local particle density is finite) describing the condensation at all. A similar situation seems to occur in the transient cases for which it is impossible to exhibit locally normal states
$\om$ describing the Bose--Einstein Condensation at mean particle density $\r(\om)$ strictly greater than the critical density $\r_c$. Indeed, it is shown that the transient cases admit locally normal states exhibiting Bose--Einstein Condensation phenomena. In order to construct such locally normal temperature states by infinite volume limits of finite volume Gibbs states, a careful choice of the the sequence of the chemical potentials should be done. For all such states, the condensate is essentially allocated on the base--point supporting the perturbation. This leads to $\r(\om)=\r_c$. We prove that all such temperature states are Kubo--Martin--Schwinger states for a natural dynamics. The construction of such a dynamics, which is a delicate issue, is also done.

\end{abstract}

\maketitle


\section{Introduction}

The present paper is devoted to the phenomena of Bose--Einstein Condensation (BEC for short) on non homogeneous 
networks obtained by adding density zero perturbations to homogeneous Cayley Trees. Even if the unperturbed model already exhibit BEC (cf. \cite{BDP}), the arising picture for the perturbed situation deserves attention as we have already explained in \cite{F}, and we are going to study in details in the sequel. The present paper is then the second part of the previous one \cite{F}, in which the harmonic analysis aspects of the model were intensively studied. The reader is also referred to \cite{FGI1} for some very interesting pivotal amenable models. We are going to investigate the physical application to the BEC of the very fascinating picture which arises from the previous mathematical investigation. The physical underlying model concerns Bardeen--Cooper Boson pairs in networks describing arrays of Josephson junctions. The formal Hamiltonian is the quartic Bose--Hubbard Hamiltonian, given on a generic network $G$ by
\begin{equation*}
H_{BH}=m\sum_{i\in VG}n_i+\sum_{i,j\in VG}A_{ij}\big(Vn_in_j-J_0a^{\dagger}_ia_j\big)\,.
\end{equation*}
Here, $VG$ denotes the set of the vertices of the network $G$, $a^{\dagger}_i$ is the Bosonic creator, and $n_i=a^{\dagger}_ia_i$ the number operator
on the site $i\in VG$ (cf. \cite{BR2}). Finally, $A$ is the adjacency operator whose matrix element
$A_{ij}$ in the place $ij$ is the number of the edges connecting the site $i$ with the site $j$ (in particular it is Hermitian). The reader is referred to the seminal paper \cite{BC} in which the theory of the superconductivity (called BCS Theory in honor of the authors) has been firstly discussed. 
It has been argued in \cite{BCRSV} that, in the case when $m$ and $V$ are negligible with respect to $J_0$, the hopping term dominates the physics of the system. After neglecting such terms, putting for the resulting effective coupling constant $J$, in principle different from $J_0$, $J=1$, and finally passing to the one--particle space, the model under consideration becomes the so--called pure hopping ones described by the {\it pure hopping one--particle Hamiltonian} which assumes the form
 \begin{equation}
\label{boha2}
 H=\|A\|\idd-A\,.
 \end{equation}
 Here, $A$ is simply the Adjacency of the fixed graph $G$, acting on the Hilbert space
 $\ell^2(VG)$. The constant $\|A\|$ is added just in order to have $H\geq0$ and does not affect the analysis at all. The reader is referred to \cite{F} for the systematic investigation of the mathematical properties arising in non amenable cases considered here consisting of negligible additive perturbations of Cayley Trees. Several interesting amenable models are treated in details in \cite{FGI1}. The reader is also referred to \cite{SRC} for some promising experiments for the pure hopping model on the Comb and Star Graphs (see Fig. \ref{figd}) pointing out an enhanced current at low temperature which might be explained by condensation phenomena.
 
One of the first attempts to investigate the BEC on non homogeneous amenable graphs, such as the Comb Graphs, was made in \cite{BCRSV}. In that paper, it was pointed out the appearance of a {\it hidden spectrum} in the low energy part of the spectrum, which is responsible of the finiteness of the critical density. In addition, the behavior of the wave function of the ground state, describing the spatial density of the condensate, was also computed. Some spectral properties of the Comb and the Star graph were also investigated in
\cite{ABO} in connection with the various  notions of independence in Quantum Probability. In that paper, it was noted the possible connection between such spectral properties and the BEC.

The systematic investigation of the BEC for the pure hopping model
on a wide class of amenable networks obtained by negligible perturbations of periodic graphs, has been started in \cite{FGI1}. The emerging results are quite surprising. First of all, the appearance of the hidden spectrum was proven for most of the graphs under consideration. This is due to the combination of two opposite phenomena arising from the perturbation. If the perturbation is sufficiently large (surprisingly, in some examples, it is enough that the perturbation is indeed finite), the norm
$\|A_p\|$ of the Adjacency of the perturbed graph becomes larger than the analogous one $\|A\|$ of the unperturbed Adjacency. On the other hand, as the perturbation is sufficiently small (i.e. zero--density), the part of the spectrum
$\s(A_p)$ in the segment $(\|A\|, \|A_p\|]$ does not contribute to the density of the
states.
This allows to compute the critical density $\r_c(\b)$ at the inverse temperature $\b$ for the perturbed model by using the integrated density of the states $F$ of the unperturbed one,
 \begin{equation}
\label{cdens01}
\r_c(\b)=\int\frac{\di F(x)}{e^{\b\left(x+(\|A_p\|-\|A\|)\right)}-1}\,.
\end{equation}
The resulting effect of the perturbed model exhibiting the hidden spectrum (i.e. when $\|A_p\|-\|A\|>0$) is that the critical density is always finite. The bridge between the Mathematics and Physics can be easily explained as follows. Consider the Bose--Gibbs occupation number (at the inverse temperature $\b>0$ and chemical potential $\m<0$) at low energies for the Hamiltonian \eqref{boha2}. After using Taylor expansion, one heuristically gets
$$
\frac1{e^{\b(H-\m I)}-I}\approx [\b(H-\m I)]^{-1}=[\b((\|A\|-\m)I-A)]^{-1}
\equiv \frac1{\b}R_A(\|A\|-\m)\,.
$$
Then the study of the BEC is reduced to the investigation of the spectral properties of the resolvent
$R_A(\l)$, for $\l\approx\|A\|$, the latter being a very familiar object for mathematicians. Due to the particular form \eqref{boha2}, as the BEC phenomena are connected with the spectral properties for small energies of the Hamiltonian, we can reduce the matter to the investigation of the spectral properties of the adjacency operator $A$ of the network close to $\|A\|$.
Another relevant fact connected with the introduction of the perturbation, and thus with the lack of homogeneity, is the possible change of the transience/recurrence character (cf. \cite{S}, Section 6) of the adjacency operator. It has to do with the possibility to construct locally normal states exhibiting BEC. As explained in \cite{FGI1}, the last relevant fact is the investigation of the shape of wave function of the ground state of the model, describing the spatial distribution of the condensate on the network in the ground state of the Hamiltonian. From the mathematical viewpoint, this is nothing else the Perron Frobenius generalized eigenvector of the Adjacency (cf. \cite{S}). It appears then clear that the physical and the mathematical aspects of the topological model based on the pure hopping Hamiltonian \eqref{boha2} are strongly related. 

Summarizing, the new and very surprising phenomena are the following. First the finiteness of the critical density is connected with the appearance of the {\it hidden spectrum} above the zero of the energy. It can be computed in a quite simple way for many amenable and non amenable models by using the {\it secular equation}, see Section 6 of \cite{FGI1} and Section 3 of \cite{F}.
This leads to the fact that the BEC can appear also in low dimensional cases $d<3$. In addition, not for all the models with finite critical density, it is possible to exhibit locally normal states (those for which the local particle density is finite in the infinite volume limit) describing BEC. Due to non homogeneity, particles condensate also in the configuration space, that is the system undergoes a "dimensional transition". This phenomenon is connected with the {\it transience/recurrence} character of the adjacency operator of the network (the last describing the  Hamiltonian of the system). Another new aspect is the difference between the geometrical dimension of the network and the growing of the wave function of the ground state. The last is called {\it Perron--Frobenius dimension} $d_{PF}$ for volume growing networks and is computed through the behavior of $\ell^2$--norms of the finite volume Perron Frobenius eigenvectors, all normalized to $1$ at a common root. If the geometrical dimension $d_{G}$ is greater than $d_{PF}$, then it is impossible to exhibit locally normal states describing the BEC, whose mean density of the particles is greater than the critical one. Surprisingly, there are non homogeneous networks for which $\r_c(\b)<+\infty$ and 
$d_{G}=d_{PF}$ (i.e. the comb graph $\bn\dashv\bz^2$, see Fig. \ref{figg}). 
\begin{figure}[ht]
     \centering
     \psfig{file=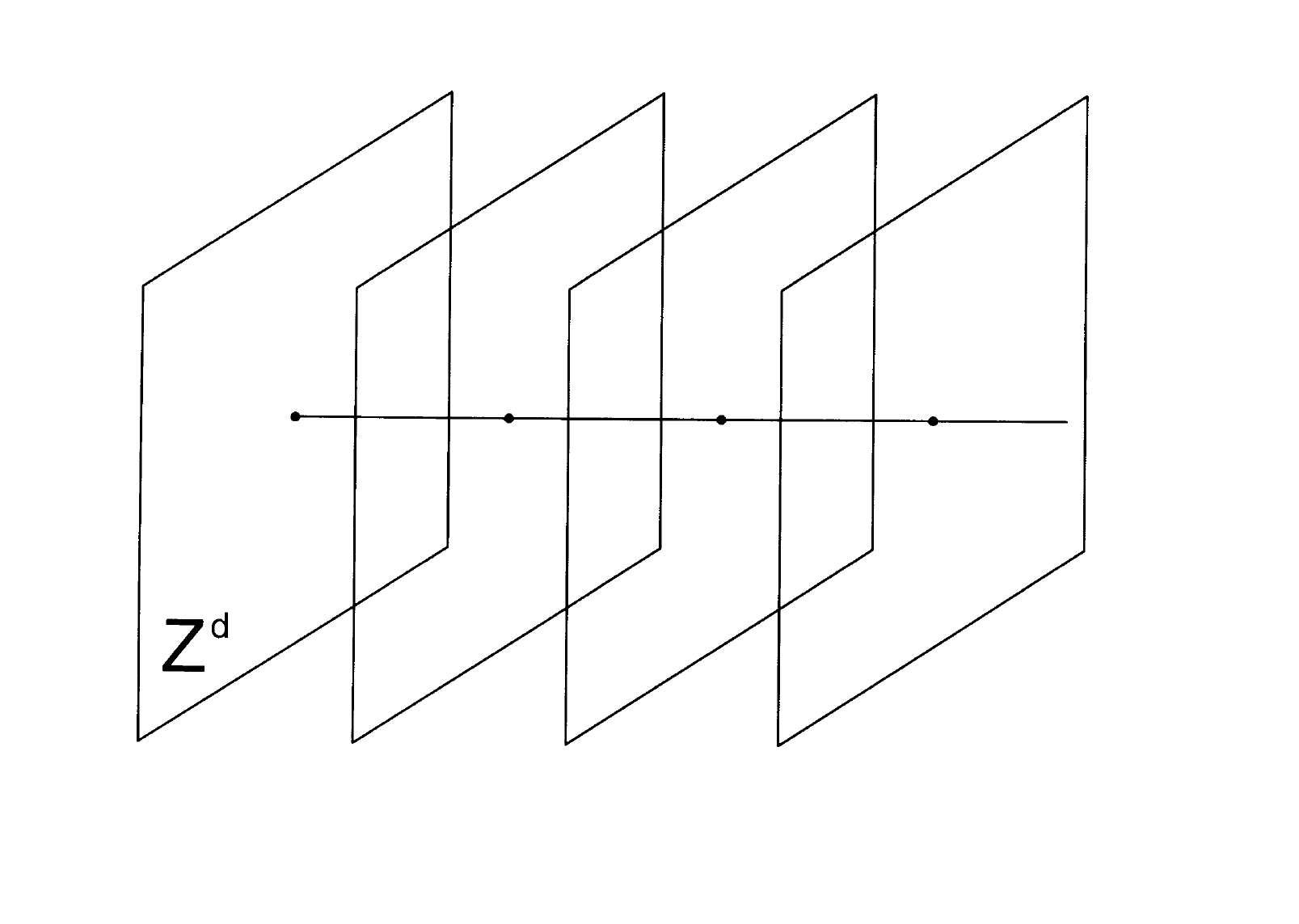,height=1.6in} \,\, \psfig{file=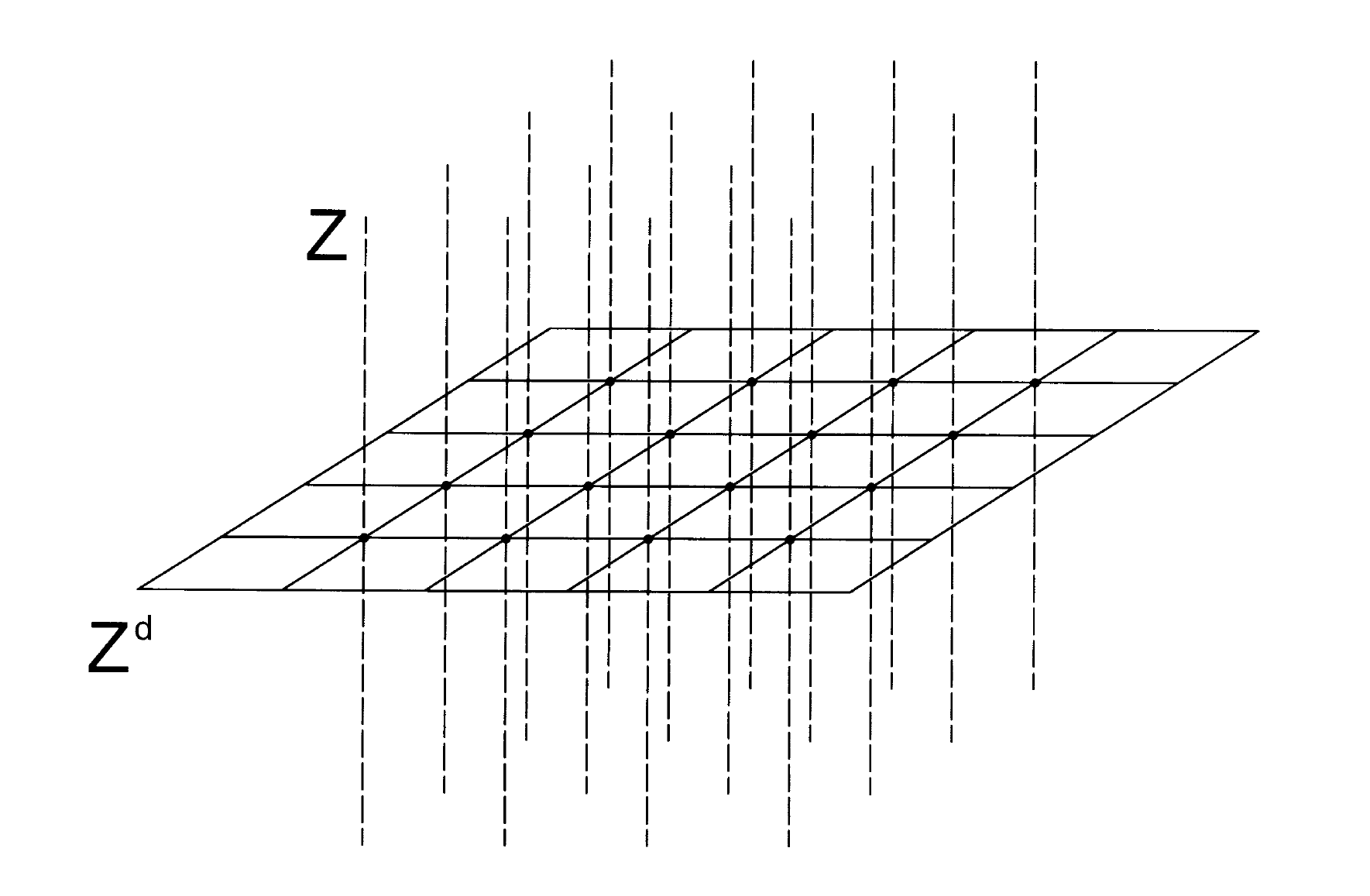,height=1.6in}
     \caption{The Comb Graphs $\bn\dashv\bz^d$ and $\bz^d\dashv\bz$.}
     \label{figg}
     \end{figure}
In addition, there are also simple non homogeneous graphs exhibiting infinite critical density, but for which it is yet possible to exhibit locally normal states describing BEC. Such a very intriguing situation is summarized as follows.\footnote{$A\dashv B$ is the comb--shaped graph whose base--point is $A$, $\r_c$ is the critical density, R/T denotes the recurrence/transience character of the Adjacency, BEC ($\r$--BEC) denotes the existence of locally normal states 
exhibiting BEC (exhibiting BEC at any mean density $\r>\r_c$).}

\vskip1cm

\begin{tabular}{||ll|c|c|c|c|c|c||}
\hline
&&  $\r_c$ & \ \ R/T \ \  & \ \ $d_G$ \ \ & \ \ $d_{PF}$ \ \ &  \ \ BEC  \ \ &  \ \ $\r$--BEC  \ \ \\
\hline

${\bz}^d, \,d<3$ &&$\infty$  &R &$d$&$d$&no&no\\ \hline 			
${\bz}^d, \,d\geq3$ &&$<\infty$  &T &$d$&$d$&yes&yes\\ \hline 
star graph&&$<\infty$  &R &$1$&$0$&no&no\\ \hline 
$\bz^d\dashv\bz, \,d<3$&&$<\infty$  &R &$d+1$&$d$&no&no\\ \hline 
$\bz^d\dashv\bz, \,d\geq3$&&$<\infty$  &T &$d+1$&$d$&yes&no\\ \hline 
$\bn$&&$\infty$  &T &$1$&$3$&yes&no\\ \hline 
$\bn\dashv\bz^2$&&$<\infty$  &T &$3$&$3$&yes&yes\\ \hline 
\end{tabular}

\vskip1cm

A similar situation also happens for some non amenable networks. 
The networks under consideration in the present paper are density zero additive perturbations of
exponentially growing graphs made of homogeneous Cayley Trees, see Fig. \ref{figa}. 
The models treated in the present paper are the perturbations $\bg^{Q,q}$,
$2\leq q<Q$, and $\bh^Q$, of the homogeneous Cayley Tree $\bg^{Q}$ along a subtree isomorphic to $\bg^{q}$, and $\bn$ respectively, see below, see also \cite{F} for the precise definition of such models and for further details. Quite surprisingly, the emerging mathematical situation described in \cite{F} for such non amenable examples, is  similar to that for the amenable graphs treated in \cite{FGI1}. Namely, we still find hidden spectrum, and the transience/recurrence character is determined by the transience/recurrence character of the negligible subgraph supporting the perturbation. Finally, the infinite volume Perron--Frobenius eigenvector is exponentially decreasing far from the base--point supporting the perturbation. Namely, concerning the wave--function
 of the ground state, we have precisely the same situation arising for the amenable situations treated in \cite{FGI1} (see Fig. \ref{figd}). 
  
The first part of the present paper is devoted to describe general properties of (bounded) Hamiltonians on networks equipped with an exhaustion such that the density of the eigenvalues of the 
underlying Hamiltonian converges in the infinite volume limit to a cumulative function 
called {\it Integrated Density of the States}. It includes as particular cases those arising by negligible perturbations of Cayley Trees as shown in \cite{F}. Then we describe the general statistical properties of such models. In particular, a careful definition of the mean density and the possible amount of the BEC condensate should be done, because of the non negligible effects of the boundary terms due to the possible lack of amenability. We pass to describe the existence of the dynamics associated to the pure hopping Hamiltonian \eqref{boha2} for the perturbed Cayley Trees. This is a delicate issue which is done in Section \ref{dynkmss} for all the cases under consideration. After that, by using the results in \cite{F}, in the transient case we write down the formula (cf. \eqref{kmsbecd}) for the natural class of states describing BEC (i.e. presenting a non negligible amount of condensate in the ground state wave--function). Such states are KMS (i.e. equilibrium) states for the dynamics previously investigated. As for the previous amenable models treated in \cite{FGI1}, such states are infinite volume KMS states whose mean density coincides with the critical one. This simply means that the condensate is allocated on a subgraph whose volume growth is negligible with respect to that of the unperturbed one. We end the paper by investigating the infinite volume behavior of finite volume Gibbs states. Due to the lack of amenability, such an investigation presents some very delicate technical difficulties arising by the unavoidable boundary effects. Modulo such difficulties, the emerging situation seems to be in accordance with that described in \cite{FGI1}, and summarized as follows. First, relatively to the general recurrent situation, it is indeed impossible to construct locally normal states exhibiting BEC at all, even if the critical density is finite. In addition, in the all transient cases under consideration, it seems to be impossible to exhibit locally normal states describing BEC at the mean density $\r>\r_c(\b)$, $\b=1/T$. Finally, in the condensation regime where for the mean density 
$\r(\om)$ of all such states 
$\r(\om)=\r_c(\b)$, a careful choice of the sequence of the finite volume chemical potentials $\m(\La_n)\to0$ should be done. By considering the geometry of the graph $G$ and the 
$\ell^2$--behavior of the ground state wave function $\left\|v\lceil_\La\right\|^2$, such a choice corresponds to the local density of the condensate
$$
C_D(\La)\approx D\frac{\left\|v\lceil_\La\right\|^2}{|\La|}\to0\,,
$$
in the infinite volume limit $\La\uparrow G$, see Proposition \ref{3a3c}. This is nothing but the naive explanation of the fact that $\r(\om)=\r_c(\b)$ for the locally normal states describing the condensation regime in all the situations under consideration.

\section{perturbed cayley trees}

A {\it graph} (called also a {\it network}) $X=(VX,EX)$ is a collection $VX$ of objects,
 called {\it vertices}, and a collection $EX$ of unordered lines connecting vertices, 
 called {\it edges}. Usually, the edges describe the physical interaction between vertices. In our model describing the pure hopping interaction, multiple edges as well as self--interactions, are also allowed. Let $E_{xy}$ be the collection of all the edges connecting $x$ with $y$. We have $E_{xy}=E_{yx}$. Let us denote by $A=[A_{xy}]_{x,y\in X}$, $x,y\in VX$, the {\it adjacency
 matrix} of $X$:
 $$
 A_{xy}=|E_{xy}|\,.
 $$ 
As explained in \cite{F}, all the geometric properties of $X$ are encoded in $A$. For example, $X$ is connected if and only if $A$ is irreducible. Setting
 $$
 \deg:=\sup_{x\in VX} \deg(x)\,,
 $$
and
$$
D_{xy}:=\deg(x)\d_{x,y}\,.
$$
We have
 $\sqrt{\deg}\leq\|A\|\leq \deg$, that is $A$ is bounded if and only if $X$
has uniformly bounded degree. For connected networks (or on each connected component) we can define the standard distance
\begin{equation}
\label{disteqa}
d(x,y):=\{\min\ell(\pi(x,y))\mid\,\pi(x,y)\,\text{a path connecting}\,x,y\}\,,
\end{equation}
$\ell(\pi)$ being the length of the path $\pi$ (i.e. the number of the edges in $\pi$).
The adjacency operator considerably differs from the Laplacian $\D=A-D$, whenever 
$X$ is non homogeneous. On the other hand, it was proven in Theorem 7.6 of \cite{FGI1} that the value of the critical density does not change under negligible perturbations on a fixed homogeneous graph, if the Hamiltonian of the model is the (opposite of the) Laplacian. Thus, the main object for the investigation of the BEC for the pure hopping model of inhomogeneous networks will be the adjacency operator.
In the present paper, all the graphs are connected, countable and with uniformly bounded
degree. In addition, we deal only with bounded operators acting on $\ell^2(VX)$ if it is not otherwise specified.

Let $B$ be a closed operator acting on
$\ell^2(VX)$, and $\l\in{\rm P(B)}\subset\bc$ the resolvent set of $B$. As usual,
$$
R_B(\l):=(\l\idd-B)^{-1}
$$
denotes the {\it Resolvent} of $B$.

Fix a bounded matrix with positive entries $B$ acting on $\ell^2(VX)$. Such an operator is called {\it positive preserving} as it preserves the elements of $\ell^2(VX)$ with positive entries. If $B$ is selfadjoint, $B$ is {\it positive} if $\langle Bu,u\rangle\geq0$ for each $u\in\ell^2(VX)$.

A sequence $\{v(x)\}_{x\in VX}$ is called a (generalized) {\it Perron--Frobenius eigenvector} if it has positive entries and
 $$
 \sum_{y\in VX}B_{xy}v(y)=\|B\|v(x)\,,\quad x\in VX\,.
 $$
If such a vector is normalizable (i.e. if it belongs to $\ell^2(VX)$) it is a standard $\ell^2(VX)$--vector, otherwise it is a weight. In all the cases we simply call Perron--Frobenius eigenvector such a sequence by dropping the word generalized.

Suppose for simplicity that $B$ is selfadjoint. It is said to be {\it recurrent} if
\begin{equation}
\label{caz}
\lim_{\l\downarrow\|B\|}\langle R_B(\l)\d_x,\d_x\rangle=+\infty\,.
\end{equation}
Otherwise $B$ is said to be {\it transient}.
It is shown in \cite{S}, Section 6, that the recurrence/transience character of $B$ does not depend on the base--point chosen in computing the limit in \eqref{caz}.
The Perron--Frobenius eigenvector is unique up to a multiplicative constant, if $X$ is finite or when $B$ is recurrent, see e.g. \cite{S}. In general, it is not unique, see the end of the present section for the case relative to the Cayley Trees.\footnote{An analogous result can be provided for the Comb Graphs in the transient situation.}
We say that an operator $B$ acting on $\ell^2(VX)$ has {\it finite
 propagation} if there exists a constant $r=r(B)>0$ such that, for any
 $x\in X$, the support of $Bv$ is contained in the (closed) ball
 $B(x,r)$ centered in $x$ and with radius $r$. It is easy to show that for the adjacency operator if $A$ on $X$, then $A^k$ has propagation $k$ for any integer $k>0$.

Let $X$ be an infinitely extended graph with an exhaustion $\{\La_n\}_{n\in\bn}$ which is kept fixed during the analysis. The definition of the 
{\it the integrated density of the states} (cf. \cite{PF}) of a selfadjoint operator $B\in\cb(\ell^2(VX))$
given in \cite{F1} works for all the situations for which it is meaningful, and then is useful for the purposes of the present work. Indeed, consider on 
$\cb(\ell^2(VX))$ the state
$$
\t_n:=\frac1{|V\La_n|}\tr_n(P_n\,{\bf\cdot}\,P_n)\,,
$$
$P_n$ being the selfadjoint projection onto $\ell^2(V\La_n)$. 
Define for a bounded operator $B$,
\begin{equation}
\label{343}
\t(B):=\lim_n\t_n(B)\,,\quad B\in\cd_\t\,,
\end{equation}
where the domain $\cd_\t$ is precisely the linear subspace of $\cb(\ell^2(VX))$ for which
the limit in \eqref{343} exists. In addition, define for a bounded selfadjoint operator $B$,
\begin{equation}
\label{3434}
\t^B(f(B)):=\lim_n\t_n(f(P_nBP_n))\,,\quad f\in C(\s(B))\,,
\end{equation}
provided such a limit exists. The domain $\cd_{\t^B}\subset C^*(B)$ is precisely the linear subspace of the unital $C^*$--algebra 
$C^*(B)\subset\cb(\ell^2(VX))$ generated by $B$, for which
the limit in \eqref{3434} exists. Notice that the definition of $\t$, $\t^B$ might depend on the chosen exhaustion which we keep fixed during the analysis.

Suppose now that $\cd_{\t^B}=C^*(B)$. Then it provides a state on $C^*(B)$ and, by the Riesz--Markov Theorem, a  Borel probability measure 
$\m_B$ on $\s(B)$. Thus, there exists a right continuous, increasing, positive function $F_B$ satisfying
$$
F_B(x)=0\,,\,\,\,x<\min\s(B)\,;\quad F_B(x)=1\,,\,\,\,x\geq\max\s(B)\,,
$$
such that
$$
\m_B((-\infty,x])=F_B(x)\,,\quad x\in\br\,.
$$
The previous described cumulative function $F_B$ is precisely the integrated density of the states (IDS for short) associated to $B$, provided it exists for the chosen exhaustion. When the graph is amenable and the operator $B$ is of finite propagation, the definition and some of the main facts relative to the IDS considerably simplify as the boundary effects play no role in the infinite volume limit, see Theorem 2.1 of \cite{F1}.

Consider the graph $Y$ such that $VY=VX$, both equipped with exhaustions
$\{X_n\}_{n\in\bn}$, $\{Y_n\}_{n\in\bn}$ such that $VY_n=VX_n$,
$n\in\bn$. The graph $Y$ is a {\it negligible} or {\it density zero
perturbation} of $X$ if it differs from $X$ by a number of edges
for which
$$
 \lim_n\frac{|\{e_{xy}\in EX\triangle EY\mid x\in VX_n\}|}{|VX_n|}=0\,,
$$
where $EX\triangle EY$ denotes the symmetric difference. To simplify
matters, we consider only perturbations involving edges, the more
general case involving also vertices can be treated analogously, see
\cite{FGI1}. Suppose that $\cd(\t^{A_X})=C^*(A_X)$ that is $A_X$ admits the IDS.
\begin{Prop} (Proposition 1.3 of \cite{F1})
\label{density0} 
Let $Y$ be a negligible perturbation of the tree $X$. Then $\cd(\t^{A_Y})=C^*(A_Y)$, and
\begin{equation*}
\t^{A_Y}(f(A_Y))=\t^{A_X}(f(A_X))\,.
\end{equation*}
\end{Prop}
If $|\l|$ is sufficiently large, then it is possible to express the Resolvent of $Y$ in terms of the Resolvent of $X$ by the Krein formula. Indeed, fix $X$ as the reference graph and define $A_X:=A$, $A_Y:=A+D$,
where $D$ is the perturbation, which eventually acts
on $\overline{\car(D)}$. Put, for $\l\in\bc$,
\begin{equation} 
 \label{s2s}
S(\l):= DPR_{A}(\l)\lceil_{\overline{\car(D)}}\,,
\end{equation}
where $P:=P_{\overline{\car(D)}}$ is the selfadjoint projection onto $\overline{\car(D)}\subset \ell^2(VX)$. The Krein formula assumes the form
 \begin{equation} 
 \label{k2s}
R_{A_Y}(\l)=R_{A_X}(\l)+R_A(\l)(\idd_{P\ell^2(VX)}-S(\l))^{-1}DPR_{A_X}(\l)\,.
\end{equation}
The resolvent formula \eqref{k2s} 
holds true for $|\l|>\|A\|+\|D\|$, and extends to any simply connected subset containing the point at infinity of $\bc\cup\{\infty\}$ made of the $\l$ for which $\idd_{P\ell^2(VX)}-S(\l)$ is invertible. Conversely, suppose that  $\idd_{P\ell^2(VX)}-S(\l)$ is not invertible for some positive number 
$\l>\|A_X\|$. Then the norm $\|A_Y\|$ of the perturbed graph $Y$ is greater than that $\|A_X\|$ of the unperturbed one $X$. It can be computed as
$$
\|A_Y\|=\sup\{\l>0\mid \idd_{P\ell^2(VX)}-S(\l)\,\,\text{is not invertible}\,\}\,.
$$
Fix now our set--up which is a graph $X$ equipped with an exhaustion $\{X_n\}_{n\in\bn}$ such that $A_X$ admits the IDS. Let $Y$ be a negligible additive perturbation of $X$ as explained above. First of all, note that Proposition \ref{density0} assures that $A_Y$ admits the IDS w.r.t. the exhaustion $\{Y_n\}_{n\in\bn}$. In addition, we obviously have $\|A^Y\|\geq\|A^X\|$. Define
\begin{equation}
\label{ddee}
\d:=\|A_X\|-\|A_Y\|\,.
\end{equation}
Denote
$F_X:=N_{\|A_X\|\idd-A_X}$, $F_Y:=N_{\|A_Y\|\idd-A_Y}$, where $N_B$ is the IDS of the operator $B$. 
\begin{Prop} (Corollary 2.6 of \cite{F}) 
\label{6}
We have
\begin{equation*}
F_Y(x)=F_X(x+\d)\,.
\end{equation*}
\end{Prop}
\begin{proof}
The IDS of $\|A_X\|\idd-A_X$ and $\|A_Y\|\idd-A_Y$ exist by hypothesis and by Proposition \ref{density0}, respectively. We get, again by Proposition \ref{density0},
\begin{align*}
&\int f(x)\di F_Y(x)=\t^{A_Y}(f(\|A_Y\|\idd-A_Y))=\t^{A_X}(f(\|A_Y\|\idd-A_X))\\
=&\t^{A_X}(f(\|A_X\|\idd-A_X-\d\idd))=\int f(x-\d)\di F_X(x)=\int f(x)\di F_X(x+\d)\,.
\end{align*}
This leads to the assertion.
\end{proof}
Suppose that $\|A^Y\|>\|A^X\|$. Then $\d<0$ in \eqref{ddee}. It has a very precise physical meaning as an effective chemical potential (cf.
\eqref{cdens01}), see
Proposition 7.1 of  \cite{FGI1}. Thus,  the part of the spectrum $(\|A_Y\|+\d,\|A_Y\|]$ does not contribute to the IDS of $A_Y$. Equivalently, the perturbed pure hopping Hamiltonian $\|A_Y\|\idd-A_Y$ presents the hidden spectrum, see Proposition \ref{hyspe}. 
The appearance of the hidden spectrum is then the combination of two different effects: the perturbation should be of density zero in order not to affect the IDS (Proposition \ref{density0}), but it should be sufficiently large in order to increase the norm of the perturbed Adjacency.

Our framework will be the investigation of the occurrence of the BEC for perturbed Cayley Trees, whose mathematical aspects are investigated in detail in \cite{F}. Namely, we consider the graphs $\bg^{Q,q}$ with $q\geq2$, and $q<Q\leq Q(q)$ where 
\begin{equation}
\label{ub}
Q(q)=\bigg[\left(2\sqrt{q-1}+1+\sqrt{4\sqrt{q-1}+1}\right)^2\bigg/4\bigg]+1
\end{equation}
($[x]$ stands for the integer part of a positive number $x$), and $\bh^{Q}$ with $2<Q\leq7$. Such networks consist of the Cayley Tree 
$\bg^{Q}$ of order $Q$, suitably perturbed along subtrees as explained in Fig. \ref{figa}. 
\begin{figure}[ht]
     \centering
     \psfig{file=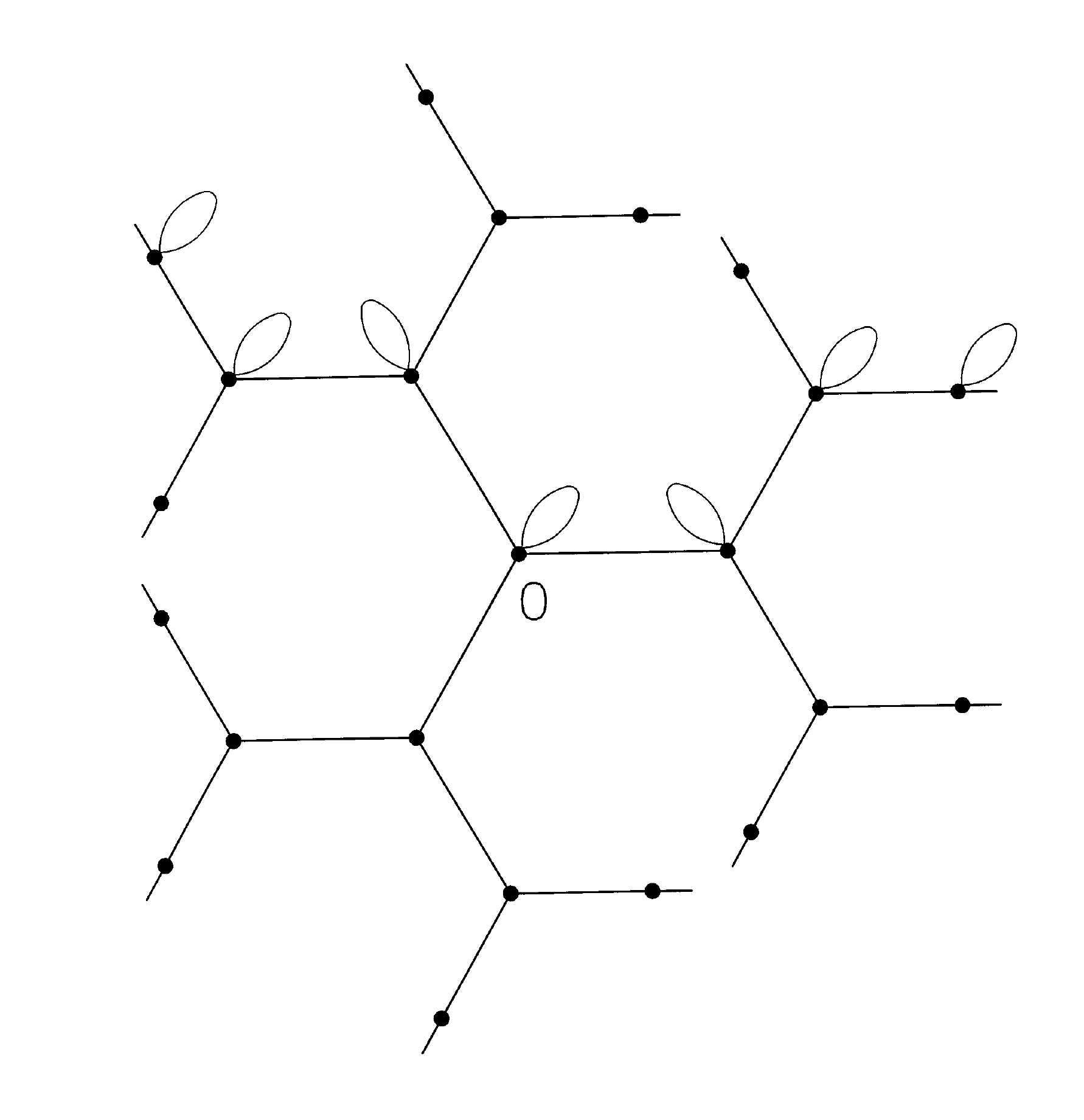,height=1.6in} \,\, \psfig{file=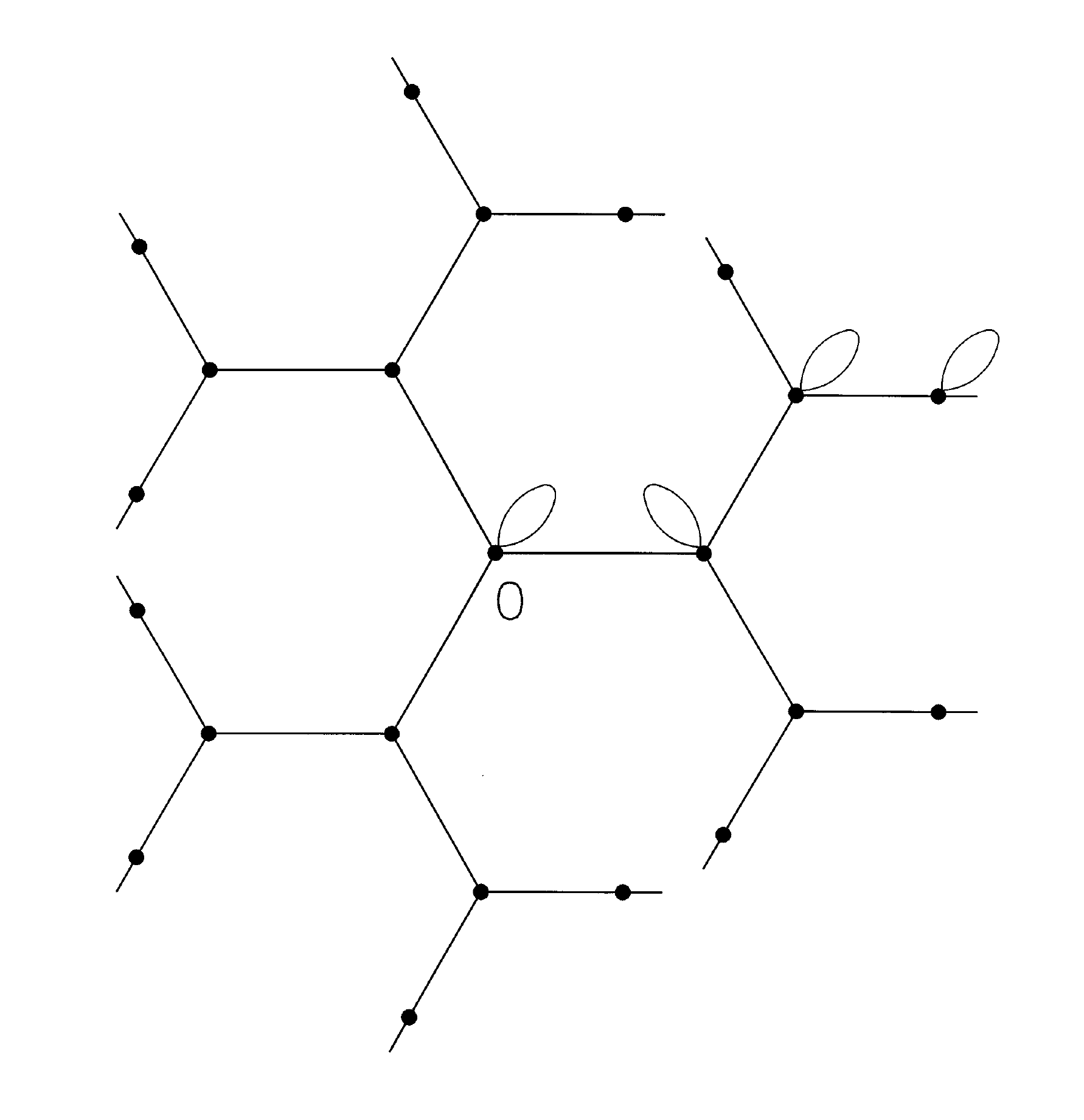,height=1.6in}\,\,\psfig{file=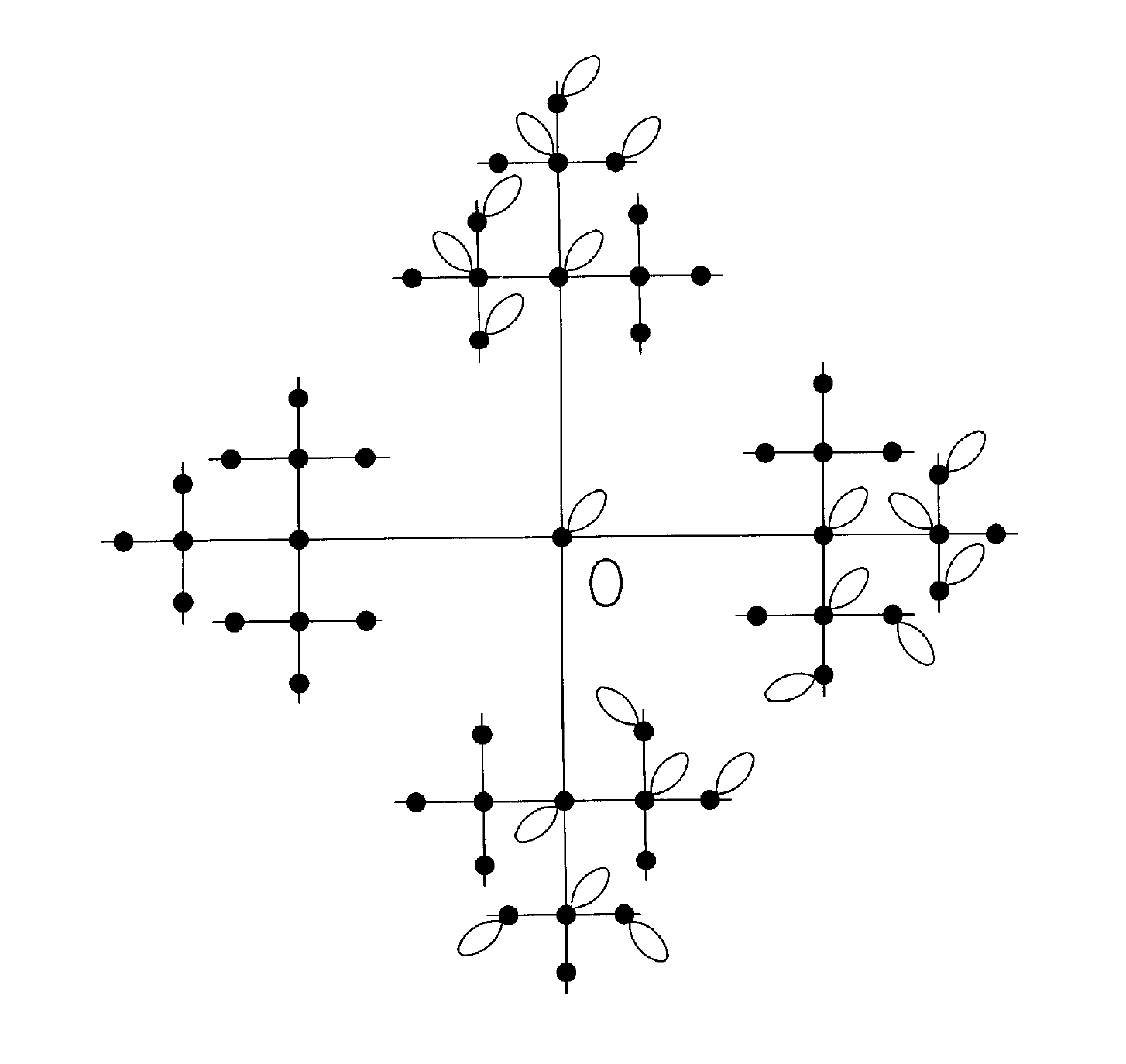,height=1.6in}
     \caption{The networks $\bg^{3,2}$, $\bh^3$ and $\bg^{4,3}$.}
     \label{figa}
     \end{figure}
One can prove (cf. Lemma 4.1 of \cite{BDP}, together with Proposition 2.2 of \cite{F}) that the Adjacency of $\bg^{Q}$ admits the IDS w.r.t. the standard exhaustion made of all the balls of radius $n$ centered in a fixed root. In the above notations, the Laplace transform $\F=L[F_{\bg^{Q}}]$ (or equivalently the {\it one--particle partition function} in physical language) of $F_{\bg^{Q}}$ is given by
\begin{equation}
\label{3}
\F(\b):=\frac{(Q-2)^2}{Q-1}\sum_{k=1}^{+\infty}\sum_{n=1}^{k}
(Q-1)^{-k}e^{-4\b\sqrt{Q-1}\sin^2\frac{n\pi}{2(k+1)}}\,.
\end{equation}
This, together with Proposition \ref{6}, leads to $F_{Y}(x)=F_{\bg^{Q}}\big(x+\|A_{\bg^{Q}}\|-\|A_{Y}\|\big)$, where $Y$ is $\bg^{Q,q}$ or $\bh^{Q}$. As $\|A_{Y}\|-\|A_{\bg^{Q}}\|>0$, they all present hidden spectrum. We also have that the graphs 
$\bg^{Q,q}$ and $\bh^{Q}$ are all transient except 
$\bg^{Q,2}$, see \cite{F}.
 
We briefly report the description of the Perron Frobenius eigenvectors $v$ given in \cite{F} for the graphs under consideration. Let $S$ be the base--point of the perturbed graph $Y$ under consideration, supporting the perturbation. The exhaustion $\{\La_n\}_{n\in\bn}$ has an obvious meaning by considering the part $S_n$ of $S$ in $B^Q_n$, the last being the ball of radius $n$ of $\bg^Q$, centered in a root $0$, kept fixed during the analysis. It was shown in \cite{F} that 
\begin{equation}
\label{3a3b}
v(x)=a(\|A_Y\|)^{d(x,S)}w(y(x))\,,\quad x\in VX\,.
\end{equation}
Here, $a(\l)$ is the function given \eqref{000}, $d$ is the distance given in \eqref{disteqa}, $y(x)$ is the unique point on $S$ such that 
$d(x,S)=d(x,y(x))$ (cf. Lemma 4.2 of \cite{F}), and finally $w$ is the Perron--Frobenius eigenvector of $PR_{A_{\bg^Q}}(\|A_Y\|)P$ 
($P\equiv P_{\ell^2(S)})$ given by
\begin{align}
\label{weign}
w(x)&=\big(1-a(\|A_{\bh^Q}\|)\big)d(x,0)+1\,,\quad x\in S\sim\bn\,,\\
w(x)&=\left(1+\frac{q-2}{q}d(x,0)\right)(q-1)^{-\frac{d(x,0)}2}\,,\quad x\in S\sim\bg^q\,,\nn
\end{align}
described in Theorems 5.2 and 6.3 of \cite{F} for $\bh^{Q}$ and $\bg^{Q,q}$, respectively.
Concerning the behavior of the $\ell^2$--norms of the restrictions of $v$, we have the following
\begin{Prop}
\label{3a3c}
For the Perron--Frobenius eigenvector $v$ described by \eqref{3a3b} and \eqref{weign}, we get
\begin{equation}
\label{1z1}
\lim_n\frac{\big\|v\lceil_{\ell^2(\La_n)}\big\|^2|S_n|}{|\La_n|}=0\,.
\end{equation}
\end{Prop}
\begin{proof}
First of all we note that, in all the situations, 
$$
v(x)=\sum_{y\in S}R_{A_{\bg^Q}}(\|A_Y\|)_{x,y}w(y)\,.
$$
This leads to 
$$
\frac{\big\|v\lceil_{\ell^2(\La_n)}\big\|^2|S_n|}{|\La_n|}
\leq(1+\big\|R_{A_{\bg^Q}}(\|A_Y\|))\big\|)^2\frac{\big\|w\lceil_{\ell^2(S_n)}\big\|^2|S_n|}{|\La_n|}\,.
$$
By taking into account the form \eqref{weign} of $w$, we easily obtain the following estimates.
Concerning $\bh^Q$ we have
$\big\|w\lceil_{\ell^2(S_n)}\big\|^2\approx n^3$, which leads to
$$
\frac{\big\|v\lceil_{\ell^2(\La_n)}\big\|^2|S_n|}{|\La_n|}\approx\frac{n^4}{Q^n}\to0\,.
$$
For the case $\bg^{Q,q}$, 
$q=2$ leads to $w(x)=1$, identically on $S$. Concerning the case $q>2$, consider the function
$$
f(\xi)=\left(1+\frac{q-2}{q}\xi\right)(q-1)^{-\frac\xi2}\,,\quad \xi\geq1\,.
$$
It is not difficult to see that $f(1)<1$ and $f'(\xi)<0$, $\xi>1$. This leads to $w(x)<1$ whenever $x\in S\backslash\{0\}$. We finally get
$$
\frac{\big\|v\lceil_{\ell^2(\La_n)}\big\|^2|S_n|}{|\La_n|}
\leq(1+\big\|R_{A_{\bg^Q}}(\|A_Y\|))\big\|)^2\frac{|B_n^q|^2}{|B_n^Q|}\approx\frac{q^{2n}}{Q^n}\to0\,.
$$
\end{proof}
We end the present section by noticing that it is easily seen that the Perron--Frobenius eigenvector is no longer unique, in general for infinitely extended networks, and in particular for the situation considered here. In fact, on $\bg^Q$, $Q>2$, $\f_{q,x_0}(x)$, $x\in\bg^Q$ are 
Perron--Frobenius eigenvectors for the Adjacency, for any fixed root $x_0\in\bg^Q$. The same happens for $\bg^{Q,q}$, $q>2$ where in 
\eqref{3a3b} it is enough to choose $w(x)=\f_{q,x_0}(x)$, $x_0\in S\sim\bg^q$  is any fixed root,  where 
$$
\f_{q,y}(x)=\left(1+\frac{q-2}{q}d(x,y)\right)(q-1)^{-\frac{d(x,y)}2}\,,\quad x,y\in \bg^q\,.
$$
The choice of the normalization at $1$ on a fixed root (which is denoted by $0$ in our framework) selects in a unique way the Perron--Frobenius vector among the class previously described. For the networks $\bg^{2}\equiv\bz$ and 
$\bg^{Q,2}$, the Perron--Frobenius eigenvector for the Adjacency is unique (up to a multiplicative constant) thanks to the fact that the Adjacency is recurrent. It is expected that it is unique also for $\bh^Q$, thanks to the fact that $S\sim\bn$ has a unique end at infinity.

\section{statistical mechanics on infinitely extended networks}

In order to investigate the statistical properties of the pure hopping model describing Bardeen Cooper pairs on arrays of Josephson junctions on non homogeneous networks we report some standard notions useful in the sequel.

Let $(\ga,\a)$ be a dynamical system consisting of a non Abelian $C^*$--algebra and a one parameter group of $*$--automorphism $\a$.\footnote{No a--priori regularity assumption is made on the automorphisms group $\{\a_t\}_{t\in\br}$ .} The  state $\f$ on the $C^{*}$--algebra 
$\gb$ satisfies the Kubo--Martin--Schwinger (KMS for short)
boundary condition at inverse temperature $\b\in\br\backslash\{0\}$ w.r.t the group of
automorphisms $\{\t_{t}\}_{t\in\br}$ if
\begin{itemize}
\item[(i)] for every $A,B\in\gb$, $t\mapsto\f(A\a_{t}(B))$, $t\mapsto\f(\a_{t}(A)B)$ are continuous;
\item[(ii)] for each $f\in\widehat{\cd}$,
\begin{equation}
\label{mmodgnss}
\int\f(A\a_{t}(B))f(t)\di t=\int\f(\a_{t}(B)A)f(t+i\b)\di t\,,
\end{equation}
where  '' $\widehat{}$ '' stands for the Fourier transform and $\cd$ is the space  of the infinitely often differentiable, compactly supported functions on $\br$.
\end{itemize}
The following facts are well known. First, a KMS is automatically invariant w.r.t. the automorphism group $\a_t$. Second, the cyclic vector 
$\Om_{\f}$ of the Gelfand--Naimark--Segal (GNS for short) 
quadruple $\big(\pi_{\f},\ch_\f,U_\f,\Om_\f\big)$ is also separating for
$\pi_{\f}(\gb)''$. Denote $\s^{\f}$ its modular group.
According to the definition of KMS boundary condition, we have
\begin{equation*}
\s^{\f}_{t}\circ\pi_{\f}=\pi_{\f}\circ\a_{-\b t}\,.
\end{equation*}
We refer the reader to \cite{BR2}, Section 5.3 for equivalent formulations of the KMS condition,  proofs and further details. 

The $C^*$--algebras considered here are those arising from the {\it Canonical Commutation Relations} (CCR for short). Namely, let $\gph$ be a pre--Hilbert space equipped with the non degenerate inner product $\langle\,{\bf\cdot}\,,\,{\bf\cdot}\,\rangle$. Denote $\bar\gph$ its completion. Consider the following (formal) relations,
\begin{equation}
\label{cccrr}
a(f)a^{\dagger}(g)-a^{\dagger}(g)a(f)=\langle g,f\rangle\,\quad f,g\in\gph\,.
\end{equation}
It is well known that the relations \eqref{cccrr} cannot be realized by bounded operators. A standard way to realize them is to look at the {\it symmetric Fock space} $\cf_+(\bar\gph)$ on which the annihilators $a(f)$ and creators $a^{\dagger}(f)$ naturally act as unbounded closed, adjoint each other (i.e. $a(f)^*=a^{\dagger}(f)$) operators. This concrete representation of the CCR is called the {\it Fock representation}. 
An equivalent description for the CCR is to put on $\cf_+(\bar\gph)$,
\begin{equation}
\label{phop}
\Phi(f):=\overline{\frac{a(f)+a^{\dagger}(f)}{\sqrt{2}}}
\end{equation}
and define the {\it Weyl operators}
$W(f):=\exp{i\Phi(f)}$. The Weyl operators are unitary and satisfy the rule
\begin{equation}
\label{cccrr1}
W(f)W(g)=e^{i\frac{\im(f,g)}{2}}W(f+g)\,,\quad f,g\in\gph\,.
\end{equation}
It is well known that the representation on $\cf_+(\bar\gph)$ of $\gw(\gph)$ described above, is faithful.
We adopt the {\it Weyl algebra} as the definition of the CCR algebra, the last being the abstract 
$C^*$--algebra $\gw(\gph)$ generated by the operators $\{W(f)\}_{f\in\gph}$ satisfying the relations \eqref{cccrr1}. Sometimes, it is also useful to compute directly the expectation values of combinations of products of annihilators and creators (equivalently field operators \eqref{phop}) in a given state.

Let $H$ be a self adjoint operator acting on $\bar\gph$. Suppose that $e^{itH}\gph\subset\gph$. Then  
the one parameter
group of automorphisms $T_tf:=e^{itH}f$ defines a one parameter group of $*$--automorphisms 
$\a_t$ of $\gw(\gph)$ by putting $\a_t(W(f)):=W(T_t f)$. The latter is called the one parameter group of {\it Bogoliubov automorphisms} generated by $T_t$.

A representation $\pi$ of the Weyl algebra $\gw(\gph)$ is {\it regular} if
the unitary group $\l\in\br\mapsto\pi(W(\l f))$ is continuous in the
strong operator topology, for any $f\in \gph$.  A state $\f$ on $\gw(\gph)$
is regular if the associated GNS representation is regular. 

The {\it quasi--free} states on the Weyl algebra are those of interest for our purposes. Such states 
$\om$ are uniquely determined by the two--point functions $\om(a^{\dagger}(f)a(g))$, $f,g\in\gph$.
A standard textbook for CCR is \cite{BR2} (cf. Section 5.2) to which the reader is referred for proofs, literature and further details.

Let $G$ be any graph. To simplify notations, we denote the sums on the vertices $VS\subset VG$ of subgraphs 
$S\subset G$, and of $VG$ itself,  directly by sums on
$S$ or $G$, respectively. 
Consider a subspace $\gph$ of $\ell^2(G)$, which contains the indicator functions $\{\d_x\mid x\in G\}$.  A representation $\pi$ of the Weyl algebra $\gw(\gph)$ is said to be {\it
locally normal (w.r.t. the Fock representation)} if
$\pi\lceil_{\gw(\ell^2(\La))}$ is quasi equivalent to the Fock
representation of $\gw(\ell^2(\La))$.  A state on $\gw(\gph)$ is said to be locally
normal if the associated GNS representation is locally normal.  A locally normal
state $\f$ does have finite local density
\begin{equation*} 
\r_{\La}(\f):=\frac{1}{|\La|}
\sum_{j\in\La}\f(a^{\dagger}(\d_{j})a(\d_{j}))
\end{equation*}
even if the mean density  might be infinite (e.g. 
$\limsup_{\La\uparrow G}\r_\La(\f)=+\infty$). 
Let $\La_{n}\uparrow G$ be an {\it exhaustion} of $G$, that is
a sequence of finite regions invading the graph $G$, together with
a sequence of states $\{\om_{\La_{n}}\}$ on $\gw(\ell^2(\La_n))$. 
\begin{Prop} (Lemma 3.2 of \cite{FGI1})
\label{add2}
 Suppose that    
\begin{equation*} 
\lim_{n}\om_{\La_{n}}(a^{\dagger}(\d_{j})a(\d_{j}))=+\infty
\end{equation*}
for some $j\in G$. 
Then $\om(W(v)):=\lim_n\om_{\La_{n}}(W(v))$ does not define any locally normal state on $\ccr(\gph)$, where $\gph\subset\ell^2(G)$ is any subspace containing the finite supported sequences.
\end{Prop}
Our purposes concerns the pure hopping model on the networks $\bg^{Q,q}$ and $\bh^Q$ or, more generally
on a uniformly bounded connected network $X$ equipped with an exhaustion $\{X_n\}$ for which the Adjacency admits the IDS (cf. \cite{F1}).
As previously explained, all the physical properties of the emerging quasi--free multi--particle model are encoded in the one--particle pure hopping Hamiltonian. Namely, the graph $X$ will be one of such networks, $Y$ a density zero additive perturbation of $X$ as described above, equipped with exhaustions $\{X_n\}_{n\in\bn}$, $\{Y_n\}_{n\in\bn}$. We denote by $G$ one of such graphs, equipped with the finite volume exhaustion $\{\La_n\}_{n\in\bn}$. The pure hopping Hamiltonian, after normalizing such that the bottom of the energy is zero, assumes the form
\begin{equation*}
\idd\|A_G\|-A_G\,.
\end{equation*}
The finite volume Hamiltonians will be 
$H_n:=P_{n}HP_{n}$ acting on $\ell^2(\La_n)$, $P_{n}\equiv P_{\ell^2(\La_n)}$ being the selfadjoint projection onto $\ell^2(\La_n)$. 

Fix a general bounded positive Hamiltonian $H\in\cb(\ell^2(G))$ admitting the IDS for a fixed exhaustion as for the pure hopping 
one.\footnote{As usual we normalize $H$ such that $0\in\s(H)$ is the bottom of the spectrum $\s(H)$ of the Hamiltonian $H$.}
Consider the finite volume densities of the states $N_n$ relative to $H_n$, and that $N$ relative to $H$. 
Define
\begin{align*}
E_{0}(H):=&\lim_{\La_{n}\uparrow G}
 \bigg(\inf\supp\big(N_n)\bigg)=0\,\\
E_{m}(H):=&\inf\supp\bigg(\lim_{\La_{n}\uparrow G}
N_n\bigg) \equiv\inf\supp\big(N\big)\,,\nn
\end{align*}
where the limit in the first equation exists by Lemma 3.4 of \cite{FGI1}, and that in the second one is meaningful directly by the definition of the IDS.
We always get $E_m(H)\geq E_0(H)=0$. 
\begin{Dfn} 
\label{hyspe}
If $E_{0}(H)< E_{m}(H)$ we say
that there is a {\it (low energy) hidden spectrum}, see e.g.
\cite{BCRSV}.
\end{Dfn}
The following assertions can be shown for all the situations under consideration. First, $H$ is continuous in $0$ (see also \cite{FGI1}, Proposition 3.9) in all the situation considered here. Second, in the unperturbed situation the hidden spectrum does not occur (by direct inspection of \eqref{3} for $\bg^Q$, or by Theorem 5.2 of \cite{FGI1} in the case of periodic amenable graphs).
If $Y$ is any perturbation of $X$, we have for the $\d$ appearing in \eqref{ddee}, 
$$
E_{0}(H_Y)-E_{m}(H_Y)=-E_m(H)=\|A_X\|-\|A_Y\|=\d\,.
$$
The opposite of $E_m(H)$ then plays the role of a chemical potential in determining the critical density of the perturbed graphs under consideration. Indeed, consider the formula for the mean density at the inverse temperature $\b$ of the particles (cf. \cite{BR2}) for the model with Hamiltonian 
$H\geq0$, and chemical potential $\m\leq0$,
\begin{equation*} 
\r_{H}(\b,\m):=\int\frac{dN_{H}(h)}{e^{\b(h-\m)}-1}\,,
\end{equation*}
and the corresponding {\it critical density}
\begin{equation*} 
\r^{H}_{c}(\b):=\int\frac{dN_{H}(h)}{e^{\b h}-1}\equiv \r_{H}(\b,0)\,.
\end{equation*}
By taking into account  Proposition \ref{6}, we get
$$
\r^{H_Y}_{c}(\b)=\r_{H_X}(\b,-E_m(H_Y))\,.
$$
This means $\r^{H_Y}_{c}(\b)\leq\r^{H_X}_{c}(\b)$. In addition, if there is hidden spectrum, $\r^{H_Y}_{c}(\b)$ is always finite and 
$\r^{H_Y}_{c}(\b)<\r^{H_X}_{c}(\b)$. It can happen that, after perturbing a graph whose critical density is infinite, we get finite critical density. As for the networks under consideration in the present paper, the unperturbed critical density is always finite (cf. \cite{BDP}), the only effect of the appearance of the hidden spectrum on the critical density is that it decreases w.r.t. the unperturbed one. 

It is well known that the condensation regime is described by $\m=0$, see e.g. Section 5.2 of \cite{BR2}, or Section 3 of \cite{FGI1}. Indeed, 
one starts from the Gibbs grand--canonical ensemble at finite volume $\La\subset G$ and fixed density. 
One fixes an inverse temperature $\b>0$ together with the particle density $\r$, and determines the
finite volume chemical potential $\m(\La)$ by solving
\begin{equation}
\label{fvden}
 \r_{H_{\La}}(\b,\m(\La))=\r\,,
 \end{equation}
 see e.g. Fig. \ref{figg1}.\footnote{It is also customary to fix the {\it  activity} $z:=e^{\b\m}$, instead of the chemical potential.}
 \begin{figure}[ht]
     \centering
     \psfig{file=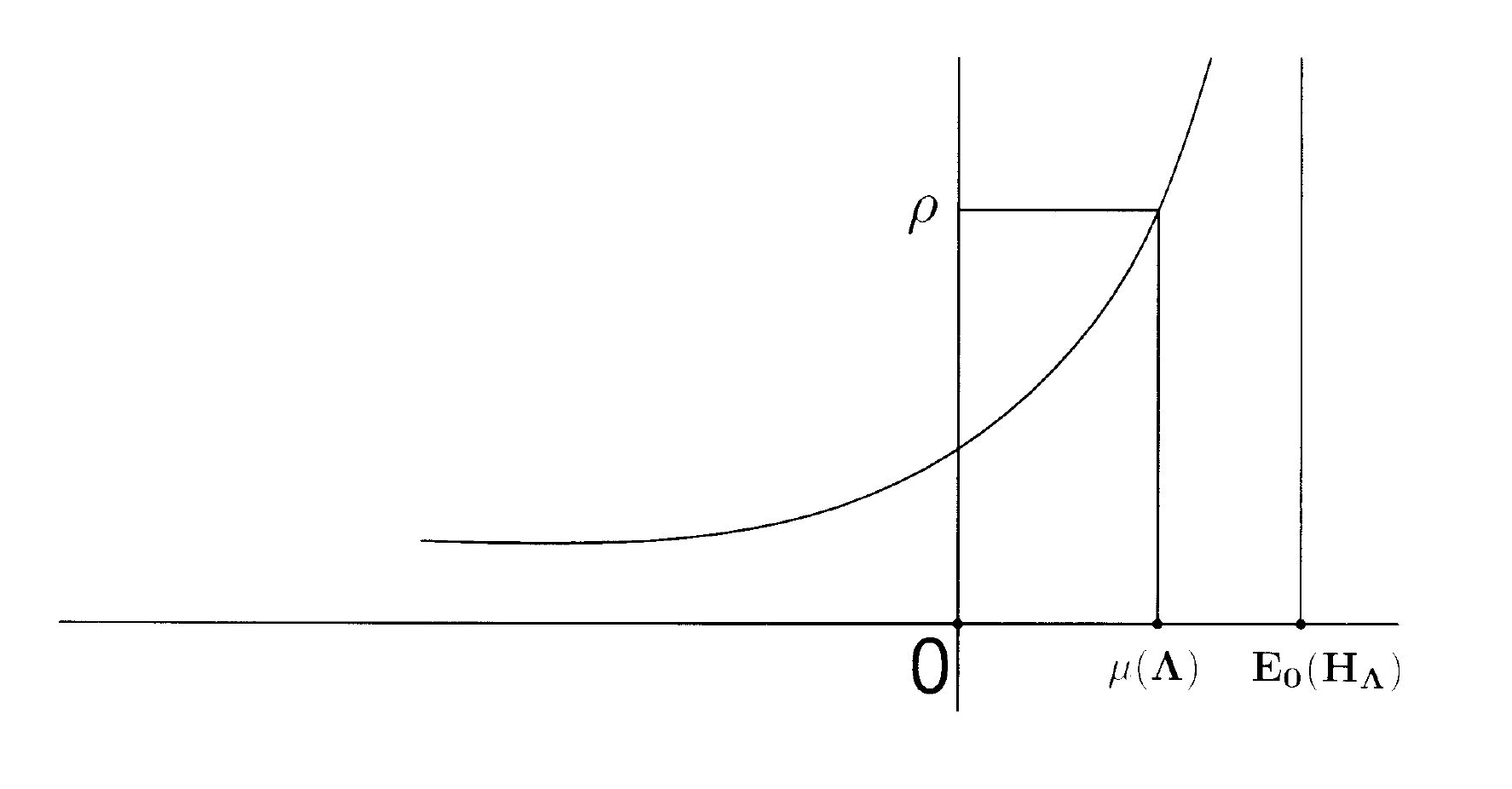,height=2.5in} 
     \caption{{}}
     \label{figg1}
     \end{figure}
The infinite volume limit can be investigated by means of the reference exhaustion 
$\{\La_n\}_{n\in\bn}$. To take into account also the very different situation appearing in non homogeneous situation, we start by fixing any general sequence of chemical potential 
$$
\m_n<\|A_G\|-\|A_{\La_n}\|\equiv E_{0}(H_{\La_n})\,,
$$
which we can suppose to converge (eventually passing to a subsequence) to some $\m$. In the case when such a sequence is recovered by using \eqref{fvden}, we put $\m_n:=\m(\La_n)$. Since $E_{0}(H_{\La_n})\to E_{0}(H)=0$, we get $\m\leq0$.
 The finite volume state with chemical density $\rho$ at $\b$ is described by
 the two point function
 \begin{equation} \label{sato}
     \om_{\La_{n}}(a^{\dagger}(\xi)a(\eta))=\big\langle  
     (e^{\b(H_{\La_{n}}-\m_{n}\idd)}-\idd)^{-1}\xi,\eta\big\rangle,
 \end{equation}
 where $\xi,\eta\in\ell^{2}(\La_n)$.
 Thermodynamical states are then
 described as limits of the finite volume states as above. 
Concerning the infinite volume limit of the finite volume density, we report a result (cf. \cite{BCRSV}), whose proof relies upon the analogous one, Proposition 3.8 of \cite{FGI1}. It describes in a rigorous way how it is possible to define the amount of the condensate in the infinite volume limit for an arbitrary sequence of finite volume chemical potential $\m_n\to0$. Due to the general character of such a result, we suppose that $H$ is a bounded positive Hamiltonian on a network $G$ equipped with an exhaustion $\La_n$ such that $H$ admits the IDS $N$ (in the situation under consideration 
$H=\|A_G\|\idd-A_G$).  Denote as usual $P_n:=P_{\ell^2(\La_n)}$, and 
$H_n:=P_nHP_n\equiv H_{\La_n}$. To avoid the almost trivial situation $\m<0$, we restrict the matter to the condensation regime $\m=0$.
\begin{Prop} 
\label{ennz}
Under the conditions listed above, the following assertions hold true for each sequence $\m_n<E_0(H_n)$ satisfying $\m_n\to0$.
\begin{itemize}
\item[(i)] Suppose that $H$ admits low energy hidden spectrum. Then for each continuous mollifier $F$ which is identically $0$ in a neighborhood of $0$, and $1$ 
in a neighborhood of $\supp N$, we have
$$
\lim_{\La_{n}\uparrow G}
\int\frac{F(x)}{e^{\b(x-\m_{n})}-1}\di N_{H_{\La_{n}}}(x)
=\r_c^H(\b)<+\infty\,.
$$
\item[(ii)] Suppose that $\lim_{x\downarrow0}N(x)=0$. Consider a sequence $\{F_\eps\}_{\eps>0}$ of increasing continuous mollifiers, all vanishing in a neighborhood of $0$ and converging almost everywhere (w.r.t the measure determined by $N$) to $1$. Then
$$
\lim_{\eps\downarrow0}\lim_{\La_{n}\uparrow G}
\int\frac{F_\eps(x)}{e^{\b(x-\m_{n})}-1}\di N_{H_{\La_{n}}}(x)
=\r_c^H(\b)\,.
$$
In addition, if $\r_c^H(\b)<+\infty$, one can replace for the continuous mollifiers 
$\{F_\eps\}_{\eps>0}$, the condition of monotony with uniform boundedness.
\end{itemize}
\end{Prop}
\begin{proof}
We start by noticing that $\frac{F(x)}{e^{\b x}-1}$, and all the $\frac{F_\eps(x)}{e^{\b x}-1}$ are continuous in a common neighborhood of all the supports of the measures determined by the cumulative functions $N_{H_n}$ and $N_H$.

(i)  By the existence of the hidden spectrum, we have
     \begin{align*}
	 &\bigg|\int\frac{F(x)}{e^{\b(x-\m_n)}-1}
	 \di N_{H_n}(x)
	 -\int\frac{ \di N_{H}(x)}{e^{\b x}-1}
	\bigg|\\
	 \leq&\bigg|\int\bigg(\frac{1}{e^{\b(x-\m_{n})}-1}
	 -\frac{1}{e^{\b x}-1}\bigg)F(x)
	 \di N_{H_n}(x)\bigg|\\
	 +&\bigg|\int\frac{F(x)}{e^{\b x}-1}\di N_{H_n}(x)
	 -\int\frac{F(x)}{e^{\b x}-1}\di N_{H}(x)\bigg|\to0,
     \end{align*}
     since
     $\frac{1}{e^{\b(x-\m_{n})}-1}\to\frac{1}{e^{\b x}-1}$,
     uniformly on the support of $F$, and the second summand
     goes to zero because $\int f(x) \di N_{H_n}(x)\to\int f(x) \di N_{H}(x)$ for each continuous function $f$
     in a common neighborhood of all the supports of the involved measures $\{\di N_n\}_{n\in\bn}$ and $\di N$.
 
 (ii) The same computation as before leads for each $\eps$,
$$     
\int\frac{F_\eps(x)}{e^{\b(x-\m_n)}-1}
	 \di N_{H_n}(x)\to
	 \int\frac{ F_\eps(x)}{e^{\b x}-1}\di N_{H}(x)\,.
$$
The proofs follows either by the Monotone or the Dominated Convergence Theorem, as 
$0=N_H(0)=\lim_{x\downarrow0}N_H(x)$.
\end{proof}
Thanks to the above result, the quantity
 $$
 n_{0}:=\lim_{\eps\downarrow0}\lim_{\La_{n}\uparrow X}
 \int\frac{1-F_{\eps}(h)}{e^{\b(h-\m_{n})}-1}dN_{H_{\La_{n}}}(h)
 $$
 is well--defined and independent of the particular choice of the
 mollifiers as above. Such a quantity describes the amount of the condensate in the ground state.  Indeed, if the sequence of chemical potentials are obtained through \eqref{fvden} by fixing the mean density $\r$, we have
 \begin{equation*}
     \r=n_{0}+\r(\b,\m)\,.
 \end{equation*}
 \begin{Rem}
 \itm{i} $\m<0$ if and only if $\r<\r_{c}(\b)$. 
  \itm{ii} For any $\b>0$, $\m<0$, the sequence \eqref{sato}
     converges pointwise to a state $\om$, whose two--point function
     is given by
     $$
     \om(a^{+}(\xi)a(\eta))=\big\langle(e^{\b(H-\m
    \idd)}-\idd)^{-1}\xi,\eta \big\rangle\,.
     $$
     Moreover, the {\it density} $\r(\om)$ of the state $\om$, defined by 
     \begin{equation*} 
	 \r(\om):=\limsup_{\La_{n}\uparrow X}\frac{1}{|\La_{n}|}
	 \sum_{j\in\La_{n}}\om_{\La_n}(a^{\dagger}(\d_{j})a(\d_{j}))\,,
     \end{equation*}
     satisfies $\r(\om)=\r(\b,\m)$.     
\itm{iii} The transience of $A$ is a necessary and sufficient condition (see Proposition \ref{add4} and Remark \ref{suftra})
for  the existence of locally normal states on $\gw(\gph)$ describing condensation phenomena, that is when for $\m=0$, for the pure hopping model.  
   \itm{vi} Observe that $n_{0}>0$ if and only if
  $\r>\r_{c}(\b)$.  Indeed, $n_0=0$ whenever $\m<0$,
  hence $n_0>0\Rightarrow \m=0$.  As a consequence,
  $0<n_0=\r-\r(\b,\m)=\r-\r_c(\b)$, i.e. $\r>\r_c(\b)$.  Conversely,
  $\r>\r_c(\b)$ implies $\m=0$, hence
  $n_0=\r-\r(\b,\m)=\r-\r_c(\b)>0$.
   \itm{v} Observe that $n_{0}>0$ can be obtained only if
  $\m_{n}\to0$.  Indeed, $\mu=0$ is a necessary condition for the
  occurrence of BEC.
 \itm{vi} Observe that $n_0>0$ can occur also when the infinite volume limit of the two--point function leads to non locally normal states i.e.
$\om_{\La_{n}}(a^{\dagger}(\d_{j})a(\d_{j}))\to+\infty$.
  \itm{vii} Observe that it is possible to have locally normal states exhibiting a non trivial amount of condensate even if $\r=\r_c(\b)$, see Section 10 of \cite{FGI1}, or Sections \ref{dynkmss}, \ref{dynkmss1} below.
 \end{Rem} 
To simplify the notations, from now on we put $\b=1$ throughout the paper if it is not otherwise specified. We also put $\r(\m):=\r^H(1,\m)$,  
$\r_{\La}(\m):=\r^{H_{\La}}(1,\m)$, and finally $\r_c:=\r^H(1,0)$, $H$ being the Hamiltonian of the model.

\section{the existence of the dynamics and the KMS states exhibiting BEC}
\label{dynkmss}
 
Let $Y$ be any of the perturbed graphs under investigation. In order to study states describing local properties, and in particular those which are locally normal, we start to consider a subspace 
$\gph\subset\ell^2(Y)$ which contains all the canonical elements $\d_j$, and which is stable under the dynamics generated by 
$H:=\|A_Y\|-A_Y$. To this end, we define
\begin{equation}
\label{hidyna}
\gph:=\ssv\left\{e^{\imath tH}\d_j\mid t\in\br\,,j\in Y\right\}
\end{equation}
for the algebraic span (no closure in $\ell^2(Y)\equiv\ell^2(\bg^Q)$) of all the elements of the form $e^{\imath tH}\d_j$. By construction $e^{\imath tH}\gph\subset\gph$. For our aims, we need also that the pairing
$h\in\gph\mapsto\langle h,v\rangle$ is well defined, when $v$ is the Perron Frobenius vector, described in \cite{F}, for the adjacency matrix $A_Y$ in each of the cases under consideration, see \eqref{3a3b} and \eqref{weign}. For $u\in\ell^2(Y)$ we also define $|u|\in\ell^2(Y)$ by putting $|u|(x):=|u(x)|$, $x\in Y$.
\begin{Lemma}
\label{rispf}
Let $f:=\sum_{n=0}^{+\infty}a_nz^n$ be a function analytic in a neighborhood of $\s(A_Y)$ such that $a_n\geq0$, $n\in\bn$. Then 
$f(A_Y)v=f(\|A_Y\|)v$.
\end{Lemma}
\begin{proof}
As the adjacency matrix has positive entries, all the involved series have positive entries. Thus, by the Monotone Convergence Theorem we get
$$
f(A_Y)v=\sum_{n=0}^{+\infty}a_nA_Y^nv=\sum_{n=0}^{+\infty}a_n\|A_Y\|^nv
=\left(\sum_{n=0}^{+\infty}a_n\|A_Y\|^n\right)v=f(\|A_Y\|)v\,.
$$
\end{proof}
The following functions, defined in $\br\backslash[-2\sqrt{Q-1},2\sqrt{Q-1}]$,
\begin{equation}
\label{000}
a(\l):=\frac{1-\sqrt{1-\frac{4(Q-1)}
{\l^2}}}{\frac{2(Q-1)}{\l}}\,,
\end{equation}
\begin{equation}
\label{333}
\m(\l):=\frac{Q-2+Q\sqrt{1-\frac{4(Q-1)}{\l^2}}}
{\frac{2(Q-1)}{\l}}
\end{equation}
are useful in the sequel. 
If $S\subset \bg^Q$, it was shown in \cite{F} that, for the situation under consideration, the operator $S(\l)$ in \eqref{s2s} describing the Resolvent of $A_Y$ through the Krein formula \eqref{k2s} assumes the form
$$
S(\l)=P_{\ell^2(S)} R_{A_{\bg^{Q}}}(\l)P_{\ell^2(S)}\,,
$$
eventually acting on $\ell^2(S)$.\footnote{With an abuse of notations, we are dropping the dependence on the subset $S$ on which is allocated the perturbation.}
In the cases under consideration when $Y$ is $\bg^{Q,q}$, $\bh^Q$, $S\sim \bg^q$ (where $\bg^2=\bz$) and $S\sim \bn$, respectively. In all such situations, \eqref{s2s} is proportional to the convolution by \eqref{000}, and for $z\in\bc\backslash[-2\sqrt{Q-1},2\sqrt{Q-1}]$ it assumes the form
\begin{equation}
\label{cvskr}
(S(z)u)(x)=\frac{\sum_{y\in S}a(z)^{d(x,y)}u(y)}{\m(z)}\,,
\end{equation}
provided at least that the series in \eqref{cvskr} converges pointwise.
Here, $d$ is the distance function given in \eqref{disteqa}, and $a(z)$ and $\m(z)$ are the analytic continuations in 
$\bc\backslash[-2\sqrt{Q-1},2\sqrt{Q-1}]$ of the corresponding objects given in \eqref{000} and \eqref{333}, respectively. It was shown in \cite{F} that \eqref{cvskr} is meaningful and defines an element of $\ell^2(\bg^Q)$, at least when $\l>\l_*\equiv\|A_Y\|$. Now we show that the series in 
\eqref{cvskr} converges and defines an element in $\ell^2(\bg^Q)$ for $z\in\bc$ such that $|z|$ is sufficiently large. 
\begin{Lemma}
\label{rispf1}
Let $a(z)$ and $\m(z)$ as above, and $r>\|A_Y\|>2\sqrt{Q-1}$. Then
$$
\G:=\max_{\{z\mid|z|=r\}}\frac{|a(z)|}{|\m(z)|}=\frac{a(r)}{\m(r)}\,.
$$
\end{Lemma}
\begin{proof}
Put $w:=2\sqrt{Q-1}/ z$. Disregarding the constants and the quantities whose modulus is fixed, we have to maximize and minimize on circles
of fixed radii less than $1$, the functions
$$
F(w):=\left|1-\sqrt{1-w^2}\right|^2\,, \quad G(w):=\left|b+\sqrt{1-w^2}\right|^2\,,
$$
where $w$ is in the open punctured disk $\Int(D_1)\backslash\{0\}$ of radius $1$ centered in the origin of the complex plane, and 
$b=\frac{Q-2}Q$, see Fig. \ref{figab}. 
\begin{figure}[ht]
     \centering
     \psfig{file=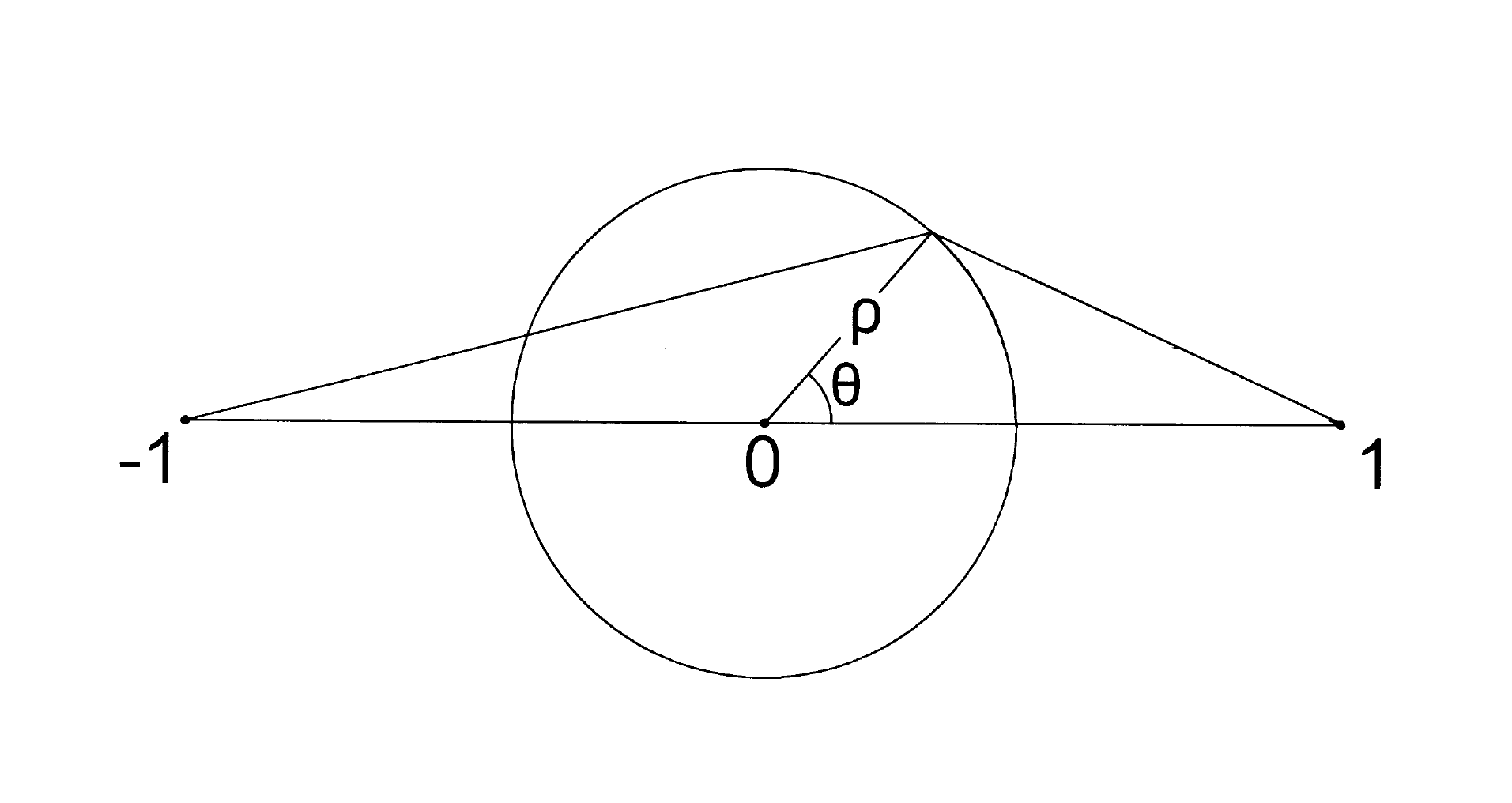,height=2.0in}
     \caption{}
     \label{figab}
     \end{figure}
First of all we note that, for the analytic continuation $\sqrt{1-w^2}$ of $\sqrt{1-x^2}$ in the cut plane
$\bc\backslash\big((-\infty,-1]\cup[1,+\infty)\big)$, we have
\begin{equation}
\label{ssyymn}
\sqrt{1-(-w)^2}=\sqrt{1-w^2}\,,\quad \sqrt{1-(\bar w)^2}=\overline{\sqrt{1-w^2}}\,.
\end{equation}
After some tedious computations, the functions $F$, $G$ are expressed by functions $f$, $g$ on the unit circle, given by
\begin{align*}
&f(\th):=1+h(\th)-\sqrt{2}\sqrt{h(\th)+1-\r^2\cos 2\th}\,, \\
&g(\th):=b^2+h(\th)+b\sqrt{2}\sqrt{h(\th)+1-\r^2\cos 2\th}\,.
\end{align*}
Here
$$
h(\th)=\sqrt{(1+\r^2)^2-4\r^2\cos^2\th}\,,
$$
$0<\r<1$ is fixed, and by \eqref{ssyymn}, we can reduce the matter to 
$0\leq\th\leq\pi/2$. After some straightforward computations, we get that
$f'(\th)$ and $g'(\th)$ are never zero for $0<\th<\pi/2$ so the the functions attain the extrema either in $0$ or in $\pi/2$. We get for $f$,
$$
\sqrt{f(0)}= 1-\sqrt{1-\r^2}>\sqrt{1+\r^2}-1=\sqrt{f(\pi/2)}\,,
$$
because $4>2(1+\sqrt{1-\r^4})$. Concerning $g$ we have,
$$
\sqrt{g(0)}=b+\sqrt{1-\r^2}< b+\sqrt{1+\r^2}=\sqrt{g(\pi/2)}\,.
$$
Collecting together, we get with $\r=2\sqrt{Q-1}/ r$,
$$
\G=\frac{1}{Q}\max_{0\leq\th\leq\pi/2}\sqrt{\frac{f(\th)}{g(\th)}}
\leq\frac{\max_{0\leq\th\leq\pi/2}\sqrt{f(\th)}}{Q\min_{0\leq\th\leq\pi/2}\sqrt{g(\th)}}=\frac{1}{Q}\sqrt{\frac{f(0)}{g(0)}}
=\frac{a(r)}{\m(r)}\leq\G
$$
and the proof follows.
\end{proof}
Denote $D_r\subset\bc$ the closed disk centered in the origin of radius $r$.
\begin{Prop}
\label{koro}
There exists $r>\|A_Y\|>2\sqrt{Q-1}$ such that for $S(z)$ in \eqref{cvskr}, we get $\|S(z)\|<1$ for each $z\in\bc\backslash D_r$. Thus, for 
$z\in\bc\backslash D_r$, the resolvent of the Adjacency of $Y$ can be written by the Neumann series
\begin{equation}
\label{koro1}
R_{A_Y}(z)=R_{A_{\bg^{Q}}}(z)
+ R_{A_{\bg^{Q}}}(z)\left(\sum_{k=0}^{+\infty}S(z)^k\right)P_{\ell^2(S)}R_{A_{\bg^{Q}}}(z)\,.
\end{equation}
\end{Prop}
\begin{proof}
Let $u=u_1-u_2+\imath(u_3-u_4)$ be the decomposition of $u\in\ell^2(S)$ in positive elements. One has $\|u_j\|\leq\|u\|$, $j=1,\dots,4$. By  taking into account \eqref{cvskr} and Lemma \ref{rispf1}, we get
\begin{align*}
&\langle S(z)u,S(z)u\rangle\leq\sum_{i,j=1}^4\langle |S(z)u_i|,|S(z)u_j|\rangle\\
\leq&\sum_{i,j=1}^4\langle S(|z|)u_i,|S(|z|)u_j\rangle
\leq16 \|S(|z|)\|^2\|u\|^2
\end{align*}
As $\|S(\l)\|$ is continuous and decreasing (cf. Lemma 3.1 in \cite{F}) with $\lim_{\l\uparrow+\infty}\|S(\l)\|=0$, there exists $r>\|A_Y\|$ such that
$\|S(r)\|\leq1/5$. Collecting together, we have in the complement of $D_r$,
$$
\|S(z)\|\leq4\|S(r)\|\leq4/5<1\,.
$$
In any simply connected open subset of $\bc\cup\{\infty\}$ containing the point at infinity for which 
$\idd_{\ell^2(S)}-S(\l)$ is invertible, $R_{A_Y}(z)$ assumes the form
$$
R_{A_Y}(z)=R_{A_{\bg^{Q}}}(z)
+ R_{A_{\bg^{Q}}}(z)R_{S(z)}(1)P_{\ell^2(S)}R_{A_{\bg^{Q}}}(z)\,,
$$
and the proof follows from the first half, by using the Neumann expansion of $R_{S(z)}(1)$.
\end{proof}
Consider the function
\begin{equation}
\label{resam}
f(x)= 
     \begin{cases}
     -\frac12\,,&x=0\,,\\
\frac{1}{e^x-1}-\frac1x\,,&x>0\,.
     \end{cases}
\end{equation}
It is bounded on $[0,+\infty)$. In addition, $\frac{1}{e^x-1}=f(x)+\frac1x$. It provides the comparison between the resolvent of $A$ and the functional calculus of $(e^H-\idd)^{-1}$ associated to the Bose--Gibbs occupation number for the pure hopping Hamiltonian.
\begin{Prop}
\label{add4}
Under the above notations, the following assertions hold true.
\begin{itemize}
\item[(i)] $A_Y$ is transient if and only if $\gph\subset\cd\left((e^H-\idd)^{-1/2}\right)$.
\item[(ii)] For $h\in\gph$ we have  $\sum_{x\in Y}|h(x)|v(x)<+\infty$.
\end{itemize}
\end{Prop}
\begin{proof}
It is enough to prove the assertions for the generators $h=e^{\imath tH}\d_j$ of $\gph$.

$(i)$ By using the splitting described by \eqref{resam} we get that
$w$ is in $\cd\left((e^H-\idd)^{-1/2}\right)$ if and only if $w$ is in $\cd(R_A(\|A\|)^{1/2})$. Let $j$ be any element of $Y$. Obviously,
$\d_j\in\gph$ and, by hypothesis, $\d_j\in\cd\left((e^H-\idd)^{-1/2}\right)$. So $\d_j\in\cd(R_A(\|A\|)^{1/2})$ and $A_Y$ is transient.
For the reverse implication, we compute for $\l>\|A_Y\|$,
\begin{align*}
\left\langle R_{A_Y}(\l)e^{\imath tH}\d_j,e^{\imath tH}\d_j\right\rangle
=&\left\langle R_{A_Y}(\l)e^{\imath t(\|A_Y\|\idd-A_Y)}\d_j,e^{\imath t(\|A_Y\|\idd-A_Y)}\d_j\right\rangle\\
=&\left\langle e^{\imath t(\|A_Y\|\idd-A_Y)}R_{A_Y}(\l)\d_j,e^{\imath t(\|A_Y\|\idd-A_Y)}\d_j\right\rangle\\
=&\left\langle R_{A_Y}(\l)\d_j,\d_j\right\rangle\uparrow\left\langle R_{A_Y}(\|A_Y\|)\d_j,\d_j\right\rangle
\end{align*}
which is finite if $A_Y$ is transient. Thus, $e^{\imath tH}\d_j\in\cd\left((e^H-\idd)^{-1/2}\right)$.

$(ii)$ 
Fix $R>r$, where $r$ is the number appearing in Proposition \eqref{cvskr}, together with the circle $C_R\subset\bc$ of radius $R$ centered in the origin. By using the Neumann expansion of $R_{S(z)}(1)$ on $C_R$ and the fact that $R_{A_{\bg^Q}}(z)$ itself is the convolution of the function $a(z)^k/\m(z)$ (cf. (7.6) in \cite{MW}), we show that the corresponding resolvent $R_{A_Y}(z)$ is a combinations of convolutions and infinite (absolutely converging) sums of powers of the convolution by the function
$a(z)^k/\m(z)$, see \eqref{koro1} and \eqref{cvskr}. Let
$C_R$ be counterclockwise oriented. Thus, thanks to Proposition \eqref{cvskr} as explained above, together with Lemma \ref{rispf}, we get
\begin{align*}
\left\langle\left|e^{\imath tH}\d_j\right|,v\right\rangle=&
\left\langle\left|\frac{1}{2\pi\imath}\oint_{C_R}e^{\imath t(\|A_Y\|\idd-z)}R_{A_Y}(z)\d_j\, \di\l
\right|,v\right\rangle\\
\leq&R\left\langle\left|R_{A_Y}(z)\d_j\right|,v\right\rangle
\leq R\left\langle R_{A_Y}(R)\d_j,v\right\rangle\\
=&R\left\langle\d_j,R_{A_Y}(R)v\right\rangle
=\frac{R\left\langle\d_j,v\right\rangle}{R-\|A_Y\|}
=\frac{Rv(j)}{R-\|A_Y\|}\,.
\end{align*}
\end{proof}
Now we pass to describe a class of locally normal states exhibiting BEC, which are indeed KMS for the natural dynamics generated on $\gw(\gph)$, $\gph$ given in \eqref{hidyna}, by the Bogoliubov automorphisms whose generator is the (second quantization of) the pure hopping Hamiltonian. To this end, we consider the dynamical system 
$(\gw(\gph),\a_t)$, where $\gw(\gph)$ is the $C^*$--algebra made of the Weyl CCR algebra on $\gph$, and $\a_t$ is generated by the Bogoliubov automorphisms $T_t u:=e^{\imath tH}u$,
$$
\a_t(W(u))=W(T_t u)\,,\quad u\in\gph\,,
$$
which are well defined as $e^{\imath tH}\gph=\gph$ by construction.

By taking into account Proposition \ref{add2}, the necessary condition for the existence of locally normal states exhibiting BEC is the transience of $A_Y$. For the pure hopping model, the transience of the Adjacency is indeed also sufficient as we are going to see. For the rest of the present section we limit our analysis when $Y$ is either $\bh^Q$, $3\leq Q\leq7$, or $\bg^{Q,q}$ where $q\geq3$, and $q<Q\leq Q(q)$ with $Q(q)$ given in \eqref{ub}, see Theorems 4.1, 5.3, 6.5 of \cite{F}, together with Proposition \ref{add2}. Fix 
$D\geq0$ and define states on the Weyl CCR algebra
$\gw(\gph)$ uniquely determined by the two--point function 
\begin{equation}
\label{kmsbecd}
\om_D(a^\dagger(u_1)a(u_2)):=\left\langle(e^H-\idd)^{-1}u_1,u_2\right\rangle+D\langle u_1,v\rangle\langle v, u_2\rangle\,,
\quad u_1,u_2\in\gph\,.
\end{equation}
The states $\om_D$ are well--defined (cf. \cite{BR2}, Section 5.2) because of Proposition \ref{add2}. In the case of quasi--free states on 
$\gw(\gph)$, it is also customary to provide the description directly in terms of their values on the Weyl unitaries. We get (cf. Section 5.2.5 of \cite{BR2}),
$$
\om_D(W(u))=e^{-\frac{\|u\|^2}{4}}e^{-\frac{\om_D(a^\dagger(u)a(u))}{2}}\,,\quad u\in\gph\,.
$$
Note that for quasi--free states, and then in our situation, the KMS condition can be checked directly in terms of $\om_D(a^\dagger(u_1)a(u_2))$, at least when 
$$
t\mapsto\om_D(a^\dagger(u_1)a(T_tu_2))
$$
is bounded and continuous. The boundedness automatically follows by the definition of the two--point function thanks to transience. The continuity can be easily checked by Lebesgue Dominated Convergence Theorem by taking into account again Proposition \ref{add2}.
\begin{Thm}
For each $D\geq0$, the states $\om_D$ given in \eqref{kmsbecd} are KMS (at inverse temperature $\b=1$) for the dynamics generated by the Bogoliubov automorphisms $\{e^{\imath Ht}\mid t\in\br\}$.
\end{Thm}
\begin{proof}
It is easily seen by Proposition \ref{add2} that $\om_D(a^\dagger(u_1)a(T_tu_2))$, $\om_D(a(T_tu_2)a^\dagger(u_1))$ are bounded and continuous (in $t$). The proof follows by an elementary change of variables in \eqref{mmodgnss} as $f\in\widehat\cd$.
\end{proof}
The second addendum in \eqref{kmsbecd} describes the amount of particles which condensates in the ground state whose wave function is nothing but the 
Perron--Frobenius weight $v$. In the finite region $\La$, such a density of the condensate is roughly given by
\begin{equation}
\label{kpf}
C_D(\La)\approx\frac1{|\La|}\sum_{x\in\La}D\langle\d_x,v\rangle\langle v, \d_x\rangle=D\frac{\left\|v\lceil_\La\right\|^2}{|\La|}\,.
\end{equation}
Even if the model exhibits a non negligible amount of the condensate, $C_D(\La)\to0$ as $\La\uparrow Y$, see Proposition \ref{3a3c}.
Formula \eqref{kpf} explains what happens in the non homogeneous pure hopping model. Particles condensate even in the configuration space. The amount of particles that the system can accommodate in the subnetwork made of the support of the perturbation is governed by the ratio between the volume growth of the original network and the $\ell^2$--growth of the wave function of the ground state. In all the situations into consideration in the present paper, $C_D(\La)$ goes to $0$ as $\La$ approaches to $Y$. This is due to the quite surprising fact that the 
Perron--Frobenius weight decreases exponentially far from the perturbed zone (cf. Fig. \ref{figd}), see \cite{F}, or also \cite{FGI1} for several amenable cases. 
\begin{figure}[ht]
     \centering
     \psfig{file=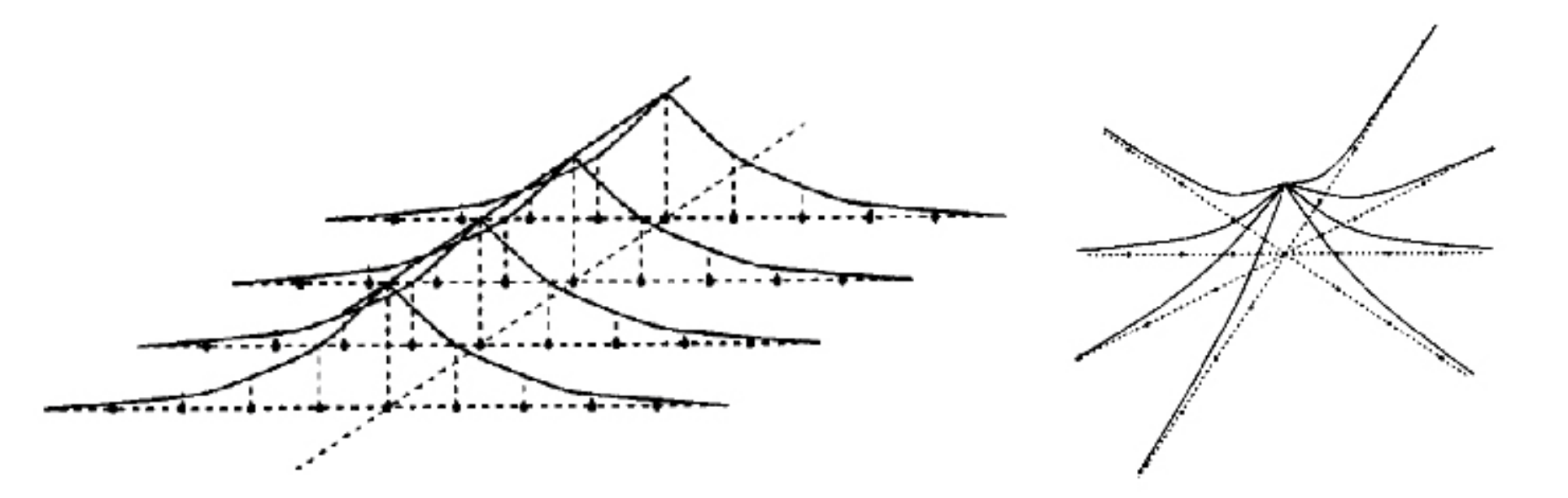,height=1.7in}
     \caption{The wave function of the ground state for the Comb and the Star Graphs.}
     \label{figd}
     \end{figure}
This is the naive explanation that the mean density of the particles associated to all the states $\om_D>0$ exhibiting BEC is the critical one 
$\r_c$. 

\section{infinite volume limits of gibbs states}
\label{dynkmss1}

We start with a simple general result explaining that, for the pure hopping model, the transience is a necessary and sufficient condition for the existence of locally normal states exhibiting BEC if one disregards the dynamics. 
\begin{Rem}
\label{suftra}
Let $G$ be a uniformly bounded degree connected network. Suppose that its Adjacency $A$ is transient. Let 
$\gph_0:=\ssv\left\{\d_x\mid x\in G\right\}$ denote 
the algebraic span of all the $\d_x$. Then there exist locally normal states on $\gw(\gph_0)$ exhibiting BEC for the pure hopping Hamiltonian
$H=\|A\|\idd-A$.
\end{Rem}
Indeed, consider any Perron--Frobenius eigenvector $u$ which exists by compactness, see e.g. Proposition 4.1 of \cite{FGI1}. Fix $a>0$ and put
\begin{equation}
\label{meanni}
\f_a(a^\dagger(u_1)a(u_2)):=\left\langle(e^H-\idd)^{-1}u_1,u_2\right\rangle+a\langle u_1,u\rangle\langle u, u_2\rangle\,,
\quad u_1,u_2\in\gph_0\,.
\end{equation}
By transience (cf. (i) of Proposition \ref{add4} which works for any network $G$ as above), \eqref{meanni} is meaningful and defines a quasi--free state on $\gw(\gph_0)$. Such a state can be interpreted as exhibiting BEC for the presence of a non trivial amount of condensate on a lowest energy wave--function $u$ described by the addendum $a\langle u_1,u\rangle\langle u, u_2\rangle$.

As explained before (see also \cite{F1}), the careful investigation of the finite volume behavior of the KMS states given in \eqref{kmsbecd} plays a relevant role because of the boundary effects which cannot be avoided in the case of non amenable graphs as in our situation. Due to such boundary effects, this analysis becomes very delicate. Yet, the emerging situation is surprisingly in accordance with previous results in 
\cite{F, FGI1}. The main aim of the present section is to investigate, technicalities permitting, the infinite volume limits of finite volume Gibbs states arising from the grand--canonical ensemble. We start by considering standard finite volume approximations of KMS states. We focus the attention to any uniformly bounded degree connected network $G$, equipped with the exhaustion $\{\La_n\}_{n\in\bn}$ made of finite volume regions. Fix a common root $0\in\La_n$, $n\in\bn$ and consider  the sequence $\{v_n\}_{n\in\bn}$ of 
Perron--Frobenius eigenvectors of the Adjacency of the finite regions $\La_n$, all normalized to $1$ at $0$. In addition, we also suppose that
$A_G$ (or equivalently the pure hopping Hamiltonian $H=\|A_G\|\idd-A_G$) admits the IDS w.r.t. the exhaustion $\{\La_n\}_{n\in\bn}$.
We reduce the matter to the $C^*$--subalgebra 
$\ga:=\gw\left(\cup_{n\in\bn}\ell^2(\La_n)\right)\equiv\gw(\gph_0)$, see Section 5.2.5 of \cite{BR2}. We denote by 
$\idd_n:=\idd\lceil_{\ell^2(\La_n)}$ the identity on $\ell^2(\La_n)$. Fix a sequence of finite volume chemical potentials $\{\m_n\}_{n\in\bn}$, with 
$$
\m_n<\|A_G\|-\|A_{\La_n}\|\,,
$$ 
converging to $\m\leq0$. The finite volume Gibbs states $\om_n$ are given, for such a sequence of chemical potential, by
\begin{equation}
\label{kmsbecd1}
\om_n(a^\dagger(u_1)a(u_2)):=\left\langle(e^{H_n-\m_n\idd_n}-\idd_n)^{-1}u_1,u_2\right\rangle\,,\quad u_1,u_2\in\gph_0\,.
\end{equation}
The case when $\m_n\to\m<0$ is quite standard (cf. \cite{BR2, BCRSV, FGI1}) and presents no further complications. Thus, we restrict ourselves to the condensation regime $\m=0$. In correspondence of such a sequence $\m_n\to0$ of chemical potential, we put
\begin{equation}
\label{lammu}
\l_n=\|A_G\|-\m_n\,.
\end{equation}
The following result takes into account what happens when the graph $G$ is recurrent.
\begin{Prop}
\label{roec}
If $G$ is recurrent then for each $x\in G$,
$$
\lim_n\om_n(a^\dagger(\d_x)a(\d_x))=+\infty\,.
$$
\end{Prop}
\begin{proof}
We consider the splitting \eqref{resam}, and $\l_n$ given in \eqref{lammu}. It is readily seen that  $\l_n>\|A_{\La_n}\|$ and $\l_n\to\|A_G\|$.
By recurrence and the monotone convergence theorem, we get
$$
\sum_{k=0}^{+\infty}\frac{\left\langle A_{G}^k\d_x,\d_x\right\rangle}{\|A_G\|^{k+1}}
=\lim_{\l\downarrow\|A_G\|}\left\langle R_{A_G}(\l)\d_x,\d_x\right\rangle=+\infty
$$
On the other hand, as $A_{\La_n}\to A_{G}$ in the strong operator topology,
$$
\frac{\left\langle A_{\La_n}^k\d_x,\d_x\right\rangle}{\l_n^{k+1}}\to\frac{\left\langle A_{G}^k\d_x,\d_x\right\rangle}{\|A_G\|^{k+1}}\,.
$$
Then, again by the monotone convergence theorem and Fatou Lemma,
\begin{align*}
&+\infty=\lim_{\l\downarrow\|A_G\|}\left\langle R_{A_G}(\l)\d_x,\d_x\right\rangle
=\sum_{k=0}^{+\infty}\frac{\left\langle A_{G}^k\d_x,\d_x\right\rangle}{\|A_G\|^{k+1}}\\
\leq&\lim_{\l_n\to\|A_G\|}\sum_{k=0}^{+\infty}\frac{\left\langle A_{\La_n}^k\d_x,\d_x\right\rangle}{\l_n^{k+1}}
=\lim_{\l_n\to\|A_G\|}\left\langle R_{\La_n}(\l)\d_x,\d_x\right\rangle
\end{align*}
Then we have just shown that the two--point function $\om_n(a^\dagger(\d_x)a(\d_x))$ diverges in the infinite volume limit.
\end{proof}
The last result together with  Proposition \ref{add2}, tells us that, even when the critical density is finite (which certainly happens in presence of hidden spectrum), it is impossible to construct locally normal states exhibiting BEC when the graph $G$ under consideration is recurrent (compare with the similar results in \cite{FGI1}). In our situation, Proposition \ref{roec} covers the case 
$\bg^{Q,2}$, with $3\leq Q\leq7$. 

Now we specialize to the situation of main interest here when the network $G$ (denoted symbolically by $Y$) is $\bh^Q$ and $\bg^{Q,q}$. In this situation, $v$ denotes the Perron--Frobenius eigenvector for $A_Y$ described in \eqref{3a3b},  
\eqref{weign}, and $v_n$ the finite volume Perron--Frobenius eigenvectors of the finite regions $\La_n$. All such Perron--Frobenius eigenvectors are normalized at $1$ on the fixed common root $0$.
\begin{Prop}
\label{faxpx111}
Let $v_n$ be the (unique) Perron--Frobenius eigenvector of $A_{\La_n}$ normalized to 1 at the common root $0\in\La_n$. Then
$\lim_n v_n(x)=v(x)$. In addition $\lim_n\|v_n\|=+\infty$.
\end{Prop}
\begin{proof}
We start by considering the networks $\S_{n}$ consisting of $\bg^Q$, perturbed along the subset  
$S_n:=S\cap B_n^Q$, see Fig. \ref{figgk}. 
\begin{figure}[ht]
     \centering
     \psfig{file=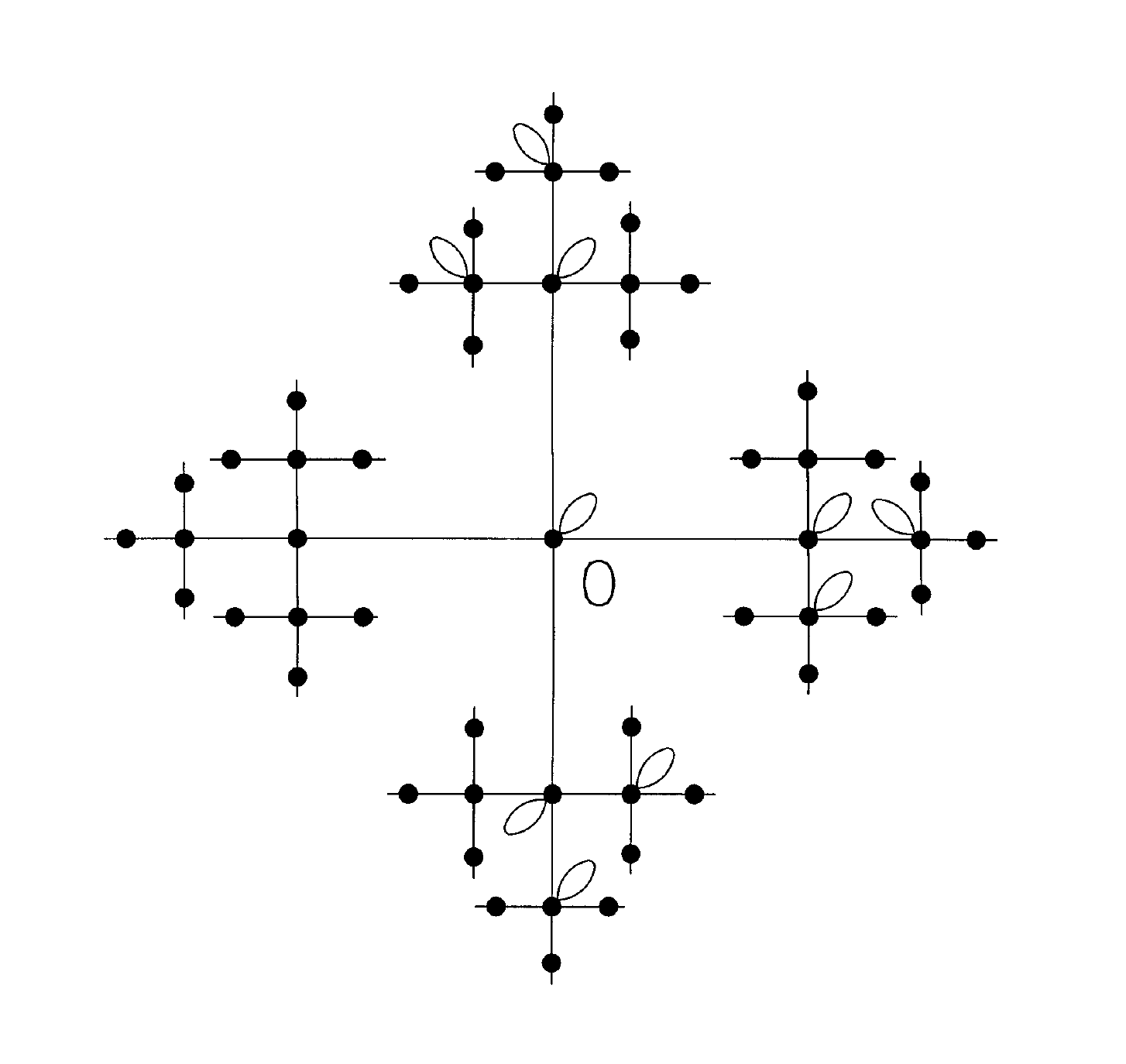,height=2.3in} \,\, \psfig{file=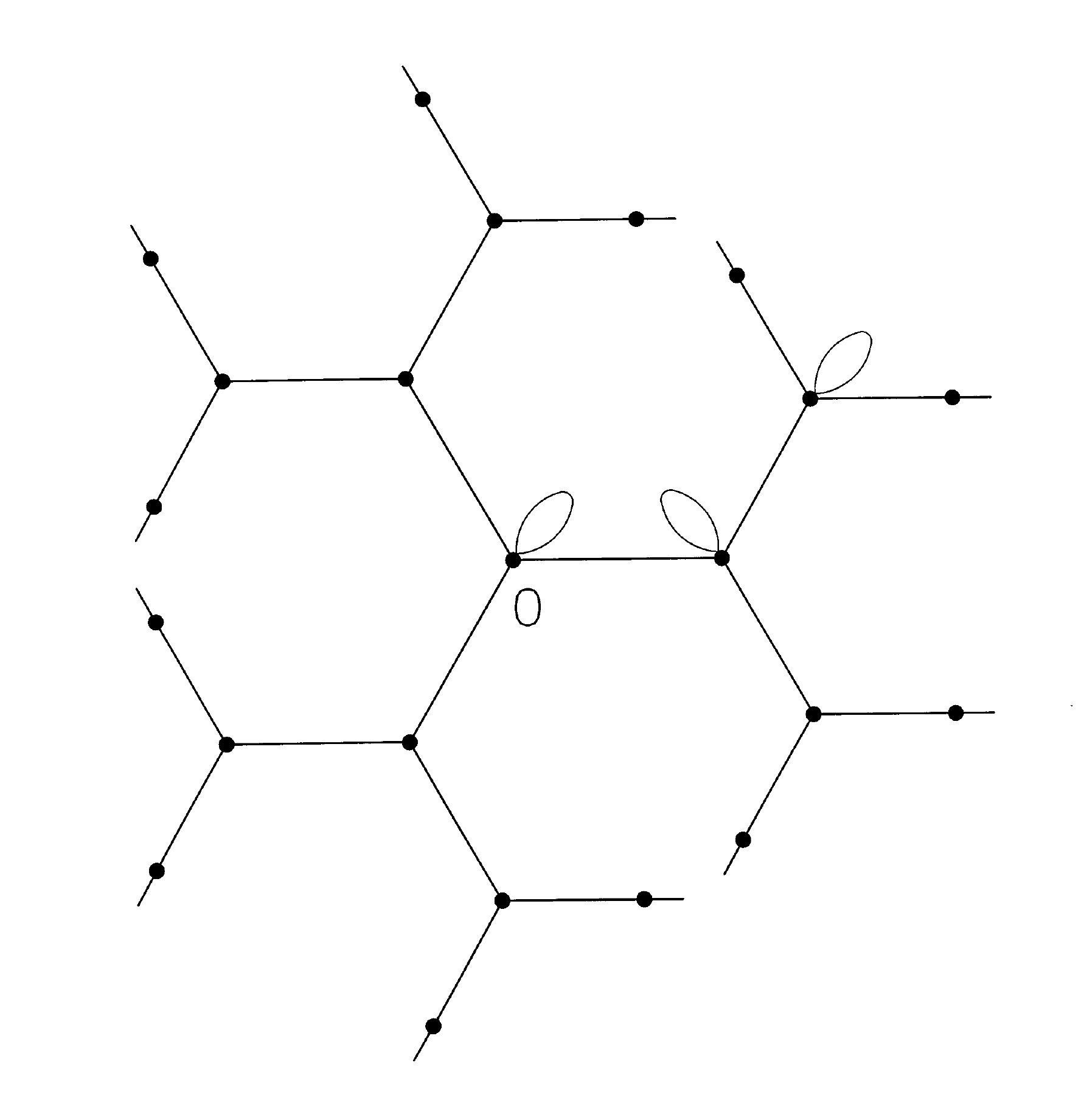,height=2.3in}
     \caption{The graphs $\S_2$ in the cases $\bg^{4,3}$ and $\bh^3$, respectively.}
     \label{figgk}
     \end{figure}The vectors $v^{(n)}$, also described in \cite{F}, are the Perron--Frobenius weights normalized at $1$ on the common fixed root. Let $r$ be any weak cluster point of $\{v_n\}_{n\in\bn}$ which exists by compactness (cf. Proposition 4.1of \cite{FGI1}). It is straightforwardly 
seen that there exists a collections $\{f_x\}_{x\in S}$ of smooth universal functions (on a sufficiently large neighborhood of $\|A_Y\|$) such that for each $x\in A_Y$ and $n>d(x,0)$,
\begin{equation*}
u(x)=a(\l)^{d(x,S)}f_{y(x)}(\l)\,,
\end{equation*}
Here, $y(x)\in S$ is the unique element such that $d(x,S)=d(x,y(x))$ (cf. Lemma 4.2 of \cite{F}),
$u$ is either $r$ or $v^{(n)}$, and, correspondingly, $\l$ is $\|A_Y\|$ or $\|A_{\S_{n}}\|$.
Notice that, it was shown in \cite{F} that
$$
\lim_n v^{(n)}(x)=v(x)\,,\quad x\in Y\,.
$$
This first implies that there exists a unique weak cluster point, say $r$, of the sequence $\{v_n\}_{n\in\bn}$. 
In addition, as $\|A_{\S_{n}}\|\uparrow\|A_{Y}\|$, we get for each $x\in Y$,
\begin{align*}
v(x)=\lim_n v^{(n)}(x)&=\lim_n a(\|A_{\S_{n}}\|)^{d(x,S)}f_{y(x)}(\|A_{\S_{n}}\|)\\
=&a(\|A_{Y}\|)^{d(x,S)}f_{y(x)}(\|A_{Y}\|)=r(x)\,.
\end{align*}
 Concerning the limits of the norms we can suppose that the $v_n$ act on $\ell^2(\bg^Q)$ by extending those to $0$ on the complement of the balls $B^Q_n$, respectively. We have by Fatou Lemma
 $$
 +\infty=\|v\|\leq\liminf_n\|v_n\|
 $$
 which completes the proof.
\end{proof}
The precise estimate of the behavior of the $\ell^2$--norms $\|v_n\|$ of the Perron--Frobenius eigenvectors $v_n$ of 
$A_{\La_n}$ as $\La_n\uparrow Y$ is needed in the sequel. Unfortunately, even if we have a precise estimate of the behavior of $\big\|v\lceil_{\ell^2(\La_n)}\big\|$, this does not guarantee the analogous estimate \eqref{1z1} for $\|v_n\|$. However, it is expected that 
\begin{equation}
\label{fregg242}
\lim_n\frac{\|v_n\|^2|S_n|}{|\La_n|}=0
\end{equation}
still holds. For the reader convenience, we report a sufficient condition, tested on simpler models, which implies \eqref{fregg242}.
Suppose that
\begin{equation*}
v_n(x)\leq v(x)\,, \quad x\in S_n\,,n\in\bn\,.
\end{equation*}
We automatically get $\big\|v_n\lceil_{\ell^2(S_n)}\big\|\leq\big\|v\lceil_{\ell^2(S_n)}\big\|$ in $\ell^2(S_n)\subset\ell^2(\bg^Q)$. 
Thus, as in the proof of Proposition \ref{3a3c}, we compute,
\begin{align*}
\frac{\|v_n\|^2|S_n|}{|\La_n|}\leq&(1+\big\|R_{A_{B_n}}(\|A_{\La_n}\|)\big\|)^2\frac{\big\|v_n\lceil_{\ell^2(S_n)}\big\|^2|S_n|}{|\La_n|}\\
\leq&(1+\big\|R_{A_{B_n}}(\|A_{\La_n}\|)\big\|)^2\frac{\big\|v\lceil_{\ell^2(S_n)}\big\|^2|S_n|}{|\La_n|}\to0
\end{align*}
because $\big\|v\lceil_{\ell^2(S_n)}\big\|^2|S_n|/|\La_n|\to0$, and $\big\|R_{A_{B_n}}(\|A_{\La_n}\|)\big\|\to\big\|R_{A_{\bg^Q}}(\|A_Y\|)\big\|$ as 
$A_{B_n}\to A_{\bg^Q}$ in the strong operator topology. The estimate \eqref{fregg242} implies the weaker one
\begin{equation}
\label{fregg2422}
\lim_n\frac{\|v_n\|^2}{|\La_n|}=0
\end{equation}
needed in Proposition \ref{pprroo}.

Now we pass to the transient situation by showing a preliminary result which, in this case, takes into account the case when the parameter $D$ describing the amount of the condensate in \eqref{kpf} on the base--point subspace is zero.
\begin{Prop}
\label{faxpx}
For each $x\in Y$ we have 
$$
\lim_n\left\langle(e^{H_n}-\idd)^{-1}\d_x,\d_x\right\rangle=\left\langle(e^{H}-\idd)^{-1}\d_x,\d_x\right\rangle\,.
$$
\end{Prop}
\begin{proof}
By using the splitting in \eqref{resam}, we can reduce the matter to the analysis of the Resolvents. Notice that
$\big\langle A_{\La_n}^k\d_x,\d_x\big\rangle<\big\langle A_{\La_{n+1}}^k\d_x,\d_x\big\rangle$. Thus, by Monotone Convergence Theorem, we get
\begin{align*}
\lim_n&\left\langle R_{A_{\La_n}}(\|A_Y\|)\d_x,\d_x\right\rangle
=\lim_n\sum_{k=0}^{+\infty}\frac{\left\langle A_{\La_n}^k\d_x,\d_x\right\rangle}{\|A_Y\|^{k+1}}\\
&=\sum_{k=0}^{+\infty}\frac{\left\langle A_{Y}^k\d_x,\d_x\right\rangle}{\|A_Y\|^{k+1}}
=\left\langle R_{A_Y}(\|A_Y\|)\d_x,\d_x\right\rangle\,.
\end{align*}
\end{proof}
We now recall some notations standardly used in the sequel. 
Put
$$
S_n(\l_n):=\idd_n-P_{\ell^2(S_n)}R_{A_{B_n}}(\l_n)P_{\ell^2(S_n)}\,.
$$
eventually acting on $\ell^2(S_n)$. Consider $B^Q_n$ the ball of radius $n$ of $\bg^Q$, $S_n:=B_n\cap S$ is the support of the perturbation inside $\La_n$ (i.e. 
$S_n\sim\bg^q$ for $\bg^{q,Q}$, and $S_n$ is isomorphic to the finite set consisting of $n+1$ points for $\bh^{Q}$). Finally, $Q_n$ is the selfadjoint projection onto $\ell^2(\La_n)\ominus\bc v_n$.  As $\La_n$ is the finite volume perturbation of $B_n^Q$, 
$\|A_{B_n^Q}\|\leq\|A_{\La_n}\|$.
\begin{Lemma}
\label{korila}
If $\l\in\big(\|A_{\La_n}\|,\G\big]$, then
$$
\frac1{1-\|S_n(\l)\|}\leq\frac{\G-\|A_{B_n}\|}{\l-\|A_{\La_n}\|}\,.
$$
\end{Lemma}
\begin{proof}
The proof immediately follows from the same computations in Lemma 3.1 of \cite{F}, by taking into account that the resolvent is differentiable and convex.
\end{proof}
The key--point to investigate the infinite volume limit for the remaining cases, that is when the underlying perturbed graph is transient, is the following
\begin{Prop}
\label{resspr}
Let $\m_n<\|A_{Y}\|-\|A_{\La_n}\|$ (i.e. $\l_n>\|A_{\La_n}\|$) such that  $\m_n\to0$ (i.e. $\l_n\to\|A_{Y}\|$). Then, 
\begin{align*}
\r_{\La_n}(\m_n)-\r_c=a_n+&\t_n\left(Q_nR_{A_{B_n}}(\l_n)R_{S_n(\l_n)}(1)R_{A_{B_n}}(\l_n)\right)\\
+&\frac1{|\La_n|(\|A_{Y}\|-\|A_{\La_n}\|-\m_n)}\,,
\end{align*}
where $a_n\to0$. In addition, if 
$$
\frac{|S_n|}{|\La_n|(\|A_{Y}\|-\|A_{\La_n}\|-\m_n)}\to0\,,
$$
then $\r_{\La_n}(\m_n)\to\r_c$.
\end{Prop}
\begin{proof}
By using the splitting \eqref{resam}, the Krein formula \eqref{k2s} and finally Lemma 2.4 of \cite{F}, we easily obtain for $\r_{\La_n}(\m_n)-\r_c$ the above formula with $a_n$ given by,
\begin{align*}
a_n=&\t_n\left(f((\|A_{Y}\|-\m_n)\idd_n-A_{\La_n})\right)-\t^{A_Y}\left(f(\|A_{Y}\|\idd_n-A_{Y})\right)\\
+&\t_n\left(R_{A_{B_n}}(\l_n)\right)-\t^{A_{\bg^Q}}\left(R_{A_{\bg^Q}}(\|A_Y\|)\right)\,,
\end{align*}
where the last goes to zero by Propositions \ref{density0} and \ref{6}.

By Lemma \ref{korila}, the second half directly follows as 
\begin{align*}
&\t_n\left(R_{A_{B_n}}(\l_n)R_{S_n(\l_n)}(1)R_{A_{B_n}}(\l_n)\right)\\
=&\t_n\left(R_{S_n(\l_n)}(1)^{1/2}R_{A_{B_n}}(\l_n)^2R_{S_n(\l_n)}(1)^{1/2}\right)\\
\leq&\|R_{A_{Y}}(\|A_{\bg^Q}\|)\|^2\t_n\left(R_{S_n(\l_n)}(1)\right)\\
\leq&\|R_{A_{Y}}(\|A_{\bg^Q}\|)\|^2\frac{|S_n|(\|A_Y\|-\|A_{B_n}\|)}{|\La_n|(\|A_{Y}\|-\|A_{\La_n}\|-\m_n)}\,.
\end{align*}
\end{proof}
Now we specialize the situation to the transient situation for the case when finite volume Gibbs states are prepared in order to describe in the infinite volume limit, a mean density $\r>\r_c$. To this end, we chose the sequence of the chemical potential by putting
\begin{equation}
\label{fregg1}
\frac1{|\La_n|(\|A_{Y}\|-\|A_{\La_n}\|-\m_n)}\to c>0\,.
\end{equation}
Even if it is expected that \eqref{fregg1} implies that the resulting limit density describes a mean density
$\r=\r_c+c$, in order to prove that we should show that
\begin{equation}
\label{fregg}
\t_n\left(Q_nR_{A_{B_n}}(\l_n)R_{S_n(\l_n)}(1)R_{A_{B_n}}(\l_n)\right)\to0\,.
\end{equation}
The last result \eqref{fregg} would follow again by the Krein formula \eqref{k2s} for the resolvent. Unfortunately, its proof relies by unavoidable boundary effects which seriously affect the related formulas and are not easily manageable to take the infinite volume limit. 
However, for the infinite volume mean density $\r$, if \eqref{fregg} holds true, we get by Proposition \ref{resspr},
$$
\r=\lim_n\r_{\La_n}(\m_n)=\r_c+\lim\frac1{|\La_n|(\|A_{Y}\|-\|A_{\La_n}\|-\m_n)}=\r_c+c>\r_c\,.
$$
This simply means that, if \eqref{fregg} holds true, the sequence \eqref{kmsbecd1} of finite volume Gibbs states obtained with the choice \eqref{fregg1} of the chemical potentials, describes a fixed mean density $\r>\r_c$ in the infinite volume limit. 
It is expected that the condition  \eqref{fregg1} is incompatible with local normality as it is shown in the following
\begin{Prop}
\label{pprroo}
Suppose that the estimate \eqref{fregg2422} holds true. If the sequence of the 
finite volume chemical potentials satisfies \eqref{fregg1}, then for the two--point function in \eqref{kmsbecd1}, we get
$$
\lim_n\om_n(a^\dagger(\d_0)a(\d_0))=+\infty\,.
$$
\end{Prop}
\begin{proof}
We have
$$
\om_n(a^\dagger(\d_x)a(\d_x))>\frac1{\|v_n\|^2(\|A_{Y}\|-\|A_{\La_n}\|-\m_n)}\,.
$$
The proof easily follows collecting together \eqref{fregg2422} and \eqref{fregg1}.
\end{proof}
The last result together with  Proposition \ref{add2}, would mean that it is impossible to construct locally normal states exhibiting BEC with mean density $\r>\r_c$ in all the situations under consideration in the present paper.

Finally, we pass for the transient cases to the construction of the states \eqref{kmsbecd} by the infinite volume limits of finite volume Gibbs states in \eqref{kmsbecd1} with a suitable choice of the sequence of chemical potentials. Also in this situation, we assume that
\begin{align}
\label{freggg}
&\lim_n\left\langle R_{A_{B_n}}(\l_n)R_{S_n(\l_n)}(1)R_{A_{B_n}}(\l_n)Q_n\d_x,Q_n\d_x\right\rangle\\
&=\left\langle R_{A_{Y}}(\|A_Y\|)R_{S(\|A_Y\|)}(1)R_{A_{Y}}(\|A_Y\|)\d_x,\d_x\right\rangle\,,x\in\bg^Q\nn
\end{align}
holds true.\footnote{By transience, it might be enough to check \eqref{freggg} only for a fixed point $x_0$ in order to obtain \eqref{freggg} itself for each other point $x\in Y$.} In order to recover \eqref{kmsbecd} by the infinite volume limits of \eqref{kmsbecd1}, we put
\begin{equation}
\label{fregg3}
\frac1{\|v_n\|^2(\|A_{Y}\|-\|A_{\La_n}\|-\m_n)}\to D\,.
\end{equation}
\begin{Prop}
Let $Y$ be either $\bh^Q$ with $3\leq Q\leq 7$, or $\bg^{Q,q}$ where $q\geq3$, $q<Q\leq Q(q)$ with 
$Q(q)$ given in \eqref{ub}. Suppose that \eqref{fregg242} holds true, and \eqref{freggg} is satisfied under the choice \eqref{fregg3}. Then for each $x\in Y$, we get for the two--point function in \eqref{kmsbecd1}, 
$$
\lim_n\om_n(a^\dagger(\d_x)a(\d_x))=\om_D(a^\dagger(\d_x)a(\d_x))\,.
$$
In addition $\r_{\La_n}(\m_n)\to\r_c$.
\end{Prop}
\begin{proof}
The second half directly follows from the assumption  \eqref{fregg242}, by the second half of Proposition \ref{resspr}. Indeed, we get, 
$$
\frac{|S_n|}{|\La_n|(\|A_{Y}\|-\|A_{\La_n}\|-\m_n)}=
\frac{|S_n|\|v_n\|^2}{|\La_n|}\frac1{\|v_n\|^2(\|A_{Y}\|-\|A_{\La_n}\|-\m_n)}\to0\,.
$$

Concerning the first part, as usual we reduce the matter to the Resolvent. As $\l_n\to\|A_Y\|$, for each fixed $\l>\|A_Y\|$, we get by monotony
$$
0<\left\langle R_{A_{Y}}(\l)\d_x,\d_x\right\rangle-\frac{v_n(x)^2}{\|v_n\|^2(\l-\|A_{Y}\|)}
\leq\left\langle R_{A_{\La_n}}(\l_n)Q_n\d_x,Q_n\d_x\right\rangle\,.
$$
By Proposition \ref{faxpx111}, we get after taking the limits in $n$ and in $\l\downarrow\|A_Y\|$,
\begin{equation*}
\left\langle R_{A_{Y}}(\|A_Y\|)\d_x,\d_x\right\rangle\leq\limsup_n\left\langle R_{A_{\La_n}}(\l_n)Q_n\d_x,Q_n\d_x\right\rangle\,.
\end{equation*}
On the other hand, by using the Krein Formula, we obtain
\begin{align*}
&\left\langle R_{A_{\La_n}}(\l_n)Q_n\d_x,Q_n\d_x\right\rangle
\leq\left\langle R_{A_{\bg^Q}}(\l_n)\d_x,\d_x\right\rangle\\
+&\left\langle R_{A_{B_n}}(\l_n)R_{S_n(\l_n)}(1)R_{A_{B_n}}(\l_n)Q_n\d_x,Q_n\d_x\right\rangle\,,
\end{align*}
which, taking into account \eqref{freggg}, leads to
\begin{align*}
&\liminf_n\left\langle R_{A_{\La_n}}(\l_n)Q_n\d_x,Q_n\d_x\right\rangle
\leq\left\langle R_{A_{\bg^Q}}(\|A_Y\|)d_x,\d_x\right\rangle\\
+&\left\langle R_{A_{Y}}(\|A_Y\|)R_{S(\|A_Y\|)}(1)R_{A_{Y}}(\|A_Y\|)\d_x,\d_x\right\rangle
=\left\langle R_{A_{Y}}(\|A_Y\|)\d_x,\d_x\right\rangle\,.
\end{align*}
Collecting together, we get the assertion.
\end{proof}
We end by noticing that the case without condensate, but at the critical point $\m=0$, is directly covered by Proposition \ref{faxpx} without any further assumption.

\section*{Acknowledgements}

The author would like to thank the University of Pretoria, and especially A. Str\"oh, for the warm hospitality and for financial support during the period March--May 2010, when the present work has been started.


\begin{thebibliography}{9999}


\bibitem{ABO} Accardi L., Ben Ghorbal A., Obata N.
{\it Monotone independence, Comb graphs and Bose--Einstein condensation},
Infin. Dimens. Anal. Quantum Probab.
Relat. Top. {\bf 7} (2004), 419--435.

\bibitem{BC} Bardeen J., Cooper L. N., Schrieffer J. R. {\it Microscopic theory of superconductivity},
Phys. Rev. {\bf 106} (1957), 162--164.

\bibitem{BDP} van den Berg M., Dorlas T. C., Priezzhev  V. B. {\it The Boson gas on a Cayley Tree},  
J. Stat. Phys. {\bf 69} (1992), 307--328

\bibitem{BR2} Bratteli O., Robinson D. W.
{\it Operator algebras and quantum statistical mechanics II},
Springer, Berlin-Heidelberg-New york, 1981.

\bibitem{BCRSV} Burioni R., Cassi D., Rasetti M., Sodano P., Vezzani A.
{\it Bose--Einstein condensation on inhomogeneous complex networks},
J. Phys. B {\bf 34} (2001), 4697--4710.

\bibitem{F} Fidaleo F. {\it Harmonic analysis on perturbed Cayley trees}, J. Func. Anal., {\bf 261} (2011), 604--634.

\bibitem{F1} Fidaleo F. {\it Corrigendum to ''Harmonic analysis on perturbed Cayley Trees''}, J. Func. Anal., {\bf 262} (2012), 4634--4637.

\bibitem{FGI1} Fidaleo F., Guido D., Isola T. 
{\it Bose Einstein condensation on inhomogeneous graphs}, 
Infin. Dimens. Anal. Quantum Probab. 
Relat. Top. {\bf 14} (2011), 149--197.





\bibitem{MW} Mohar B., Woess W. {\it A survey on spectra of infinite graphs}, Bull. London Math. Soc.  {\bf 21} (1982), 209--234.

\bibitem{PF} Pastur L., Figotin A. {\it Spectra of random and almost--periodic
operators}, Springer-Verlag, Berlin, 1992.

\bibitem{S} Seneta E. {\it Nonnegative matrices and Markov chains}. Springer--Verlag, Berlin, Heidelberg New York, 1981.


\bibitem{SRC} Silvestrini P., Russo R., Corato V., Ruggiero B., Granata C., Rombetto S.,
Russo M., Cirillo M., Trombettoni A., Sodano P. {\it Topology--induced critical current enhancement in Josephson networks},
Phys. Lett. A {\bf 370} (2007), 499--503.



\end{thebibliography}
\end{document}